\documentclass[letterpaper,11pt]{article} 

\usepackage{comment}

\usepackage{amsmath, amsfonts, amssymb, amsthm, amscd, graphicx, 
enumerate, footmisc, mathtools, xcolor}
\providecommand{\U}[1]{\protect\rule{.1in}{.1in}}
\def\theenumi{\arabic{enumi}}

\def\theenumii{\alph{enumii}}
\def\p@enumii{\theenumi.}

\def\theenumiii{\arabic{enumiii}}
\def\p@enumiii{(\theenumi)(\theenumii)}

\def\p@enumiv{\p@enumiii.\theenumiii}
\newenvironment{lcases}
  {\left\lbrace\begin{aligned}}
  {\end{aligned}\right.}
\newcommand{\bbF}{\mathbb F}

\newcommand{\bbZ}{\mathbb Z}

\newcommand{\Char}{\operatorname{char}}

\newcommand{\rank}{\operatorname{rank}}
\newcommand{\symrank}{\operatorname{rk_\mathsf S}}

\newcommand{\im}{\operatorname{im}}
\newcommand{\Span}{\operatorname{span}}
\newcommand{\Aut}{\operatorname{Aut}}

\newcommand{\calPLG}{\mathcal{PLG}}
\newcommand{\calG}{\mathcal G}
\newcommand{\calK}{\mathcal K}

\newcommand{\bfx}{\mathbf x}
\newcommand{\bfy}{\mathbf y}

\newcommand{\bfv}{\mathbf v}
\newcommand{\bff}{\mathbf f}

\newcommand{\calI}{\mathcal I}

\usepackage{fullpage}
\usepackage{mathrsfs}
\newtheorem{theorem}{Theorem}[section]

\newtheorem{corollary}[theorem]{Corollary}

\newtheorem{lemma}[theorem]{Lemma}

\numberwithin{equation}{section}

\usepackage{titling}
\usepackage{enumitem} 
\setlist{noitemsep,topsep=0pt,parsep=0pt} 
\usepackage[margin=1cm]{caption} 
\usepackage{subfig}
\usepackage{tikz}

\numberwithin{equation}{section}
\usepackage[bookmarks=true,hypertexnames=false]{hyperref}
\hypersetup{colorlinks=true, citecolor=blue, linkcolor=red, urlcolor=blue}
\makeatletter

\newcommand{\Rmnum}[1]{\expandafter\@slowromancap\romannumeral #1@}
\makeatother

\newcommand{\Sym}{\operatorname{Sym}}

\newenvironment{remark}{\medskip{\bfseries \noindent Remark:}}{\par\medskip}{\par\medskip}

\newcommand{\simp}{\operatorname{simp}}
\newcommand{\orb}{\operatorname{orb}}
\newcommand{\Isom}{\operatorname{Isom}}

\usepackage[export]{adjustbox} 
\usepackage[normalem]{ulem}

\begin{document}

\title{\textbf{On a Theorem of Lov\'asz that $\hom(\cdot, H)$ Determines the Isomorphism Type of $H$}\thanks{A preliminary version of this paper appeared in the Proceedings of ITCS 2020~\cite{Cai-Govorov-GH-hom-type}.}} 

\vspace{0.3in}

\author{Jin-Yi Cai\thanks{Department of Computer Sciences, University of Wisconsin-Madison. Supported by 
NSF CCF-1714275.} \\ {\tt jyc@cs.wisc.edu}
\and Artem Govorov\thanks{Department of Computer Sciences, University of Wisconsin-Madison. Supported by 
NSF CCF-1714275. 
 } \thanks{Artem Govorov is the author's preferred spelling 
of his name,
rather than the official spelling Artsiom Hovarau.
 } \\ {\tt hovarau@cs.wisc.edu}}

\date{}
\maketitle

\bibliographystyle{plainurl}

\begin{abstract}
Graph homomorphism has been an important research topic since
its introduction~\cite{Lovasz-1967}. Stated in the language of 
binary relational structures in that paper~\cite{Lovasz-1967}, Lov\'asz 
 proved  a fundamental theorem
that, for a graph $H$ given by its $0$-$1$ valued adjacency matrix,
 the graph homomorphism function $G \mapsto \hom(G, H)$
determines the isomorphism type of $H$. In the past 50 years various
extensions have been proved by
many researchers~\cite{Lovasz-2006,Freedman-Lovasz-Schrijver-2007,
Borgs-et-al,Schrijver-2009,Lovasz-Szegedy-2009}.
These extend the basic $0$-$1$ case to admit vertex  and edge weights; 
but these extensions all have
some restrictions such as  all  vertex weights must be positive.
In this paper we prove a general form of this
theorem where $H$ can have arbitrary vertex and edge weights.
A noteworthy aspect is that we prove this by a surprisingly simple and
unified argument.
This bypasses various technical obstacles
and unifies and extends all previous known versions of this theorem
on graphs.
The constructive proof of 
our theorem can be used to make various complexity dichotomy theorems
for graph homomorphism \emph{effective} in the following sense:
it provides an algorithm
that for any $H$ either outputs a P-time algorithm
solving $\hom(\cdot, H)$
or a P-time reduction from a canonical \#P-hard problem
to  $\hom(\cdot, H)$.
\end{abstract}

\thispagestyle{empty}
\clearpage
\setcounter{page}{1}

\section{Introduction}\label{sec:intro}

More than 50 years ago the concept of graph homomorphism 
was introduced~\cite{Lovasz-1967,GH-book}.
Given two graphs $G$ and $H$, a mapping from $V(G)$ to $V(H)$ is called
a homomorphism if every edge of $G$ is mapped to an edge of $H$.
The graphs $G$ and $H$ can be either both directed or undirected.
Presented in the language of binary relational structures,
Lov\'asz proved in that paper~\cite{Lovasz-1967}
 the following fundamental theorem about graph homomorphism:
 If $H$ and $H'$ are two graphs,
then they are isomorphic iff they define the same counting
graph homomorphism function, namely, for every $G$, the number
of homomorphisms from $G$ to $H$ is the same as
that from $G$ to $H'$.  This number is denoted by $\hom(G , H)$.
(Formal definitions are in Section~\ref{sec:preliminaries}.)

In~\cite{Lovasz-1967} the graph $H$ is a $0$-$1$ adjacency matrix;
there are no vertex and edge weights.
In~\cite{Freedman-Lovasz-Schrijver-2007} Freedman, Lov\'asz and Schrijver 
define a weighted version of the homomorphism function $\hom(\cdot, H)$,
where $H$ has \emph{positive} vertex weights and \emph{real} edge weights.
The paper~\cite{Freedman-Lovasz-Schrijver-2007} investigates
what graph properties can be expressed
as such graph homomorphism functions. They gave a necessary and
sufficient condition for this expressibility. 
This work has been extended to
the case with arbitrary vertex and edge weights in
a field~\cite{Cai-Govorov-2019}, and to
``edge coloring models'', 
e.g.,~\cite{B-Szegedy,Regts-2013}. A main technical tool
introduced in~\cite{Freedman-Lovasz-Schrijver-2007} is the so-called
graph algebras. In~\cite{Lovasz-2006} Lov\'asz further investigated
these graph algebras and proved precise bounds for their dimensions.~\footnote{See
also~\cite{Lovasz-loops-correction} for a small correction suggested by Martin 
Dyer.
However, this fix
works only for
undirected weighted graphs with positive vertex and real edge weights.
} 
These dimensions are a quantitative account of the space of
all isomorphisms from $H$ to $H'$. They are expressed in a theory
of labeled graphs.
Schrijver~\cite{Schrijver-2009}
studied the function $\hom(\cdot, H)$ where $H$ is an undirected graph 
with complex edge weights (but all vertex weights are restricted to $1$).  
He also gave a characterization of a graph property expressible
in this form, and proved that $\hom(\cdot, H) = \hom(\cdot, H')$ implies that
$H \cong H'$ for undirected graphs with complex edge weights (but
unit vertex weights). See also~\cite{Schrijver-2013}, where 
edge weights are complex but  vertex weights are restricted to $1$.
Regts in~\cite{Regts-2013}, in addition
to finding interesting connections between edge coloring
models and invariants of the orthogonal group,
also proved multiple theorems in the framework
of graph homomorphisms (corresponding to ``vertex models'') requiring that
all (nonempty) sums of vertex weights be nonzero.
The possibility
that vertex weights may sum to zero has been a difficult point.
Our main result is to extend this isomorphism theorem to both directed and 
undirected graphs with arbitrary vertex and edge weights.
We also determine the precise values of the dimensions of the 
corresponding graph algebras. A variant of our result, 
in terms of dimensions of associated algebras, was proved by Goodall,
Regts and Vena~\cite{Goodall-et-al}; see the corresponding remark in Section~\ref{sec:MainProof}.

To prove our theorem,
we introduce a surprisingly simple and completely elementary argument, which
we call the \emph{Vandermonde Argument}.
All of our results are proved by this one technique.

Two vertices $i$ and $j$ in an unweighted graph $H$ are called
\emph{twins} iff the neighbor sets of $i$ and $j$ are identical.
For weighted graphs, $i$ and $j$ are called twins iff
the edge weights $\beta(i,k) = \beta(j, k)$ 
(and for directed graphs also $\beta(k,i) = \beta(k,j)$) for all $k$.
In order to identify the isomorphism class of $H$, 
a natural step 
is to combine twin vertices. This creates a
super vertex with a combined vertex weight
(even when originally all                
vertices are unweighted, i.e., have weight $1$).
After this ``twin reduction'' step, our isomorphism theorem
can be stated. The following is a simplified form.

\begin{theorem}\label{thm:intro-gh-both-twin-free-k=0}
Let $\bbF$ be a field of characteristic $0$.
Let $H$ and $H'$ be (directed or undirected) weighted graphs 
with arbitrary vertex and edge weights from $\bbF$.
Without loss of generality all individual vertex weights are nonzero.
Suppose $H$ and $H'$ are twin-free.
If for all simple graphs $G$ (i.e., without loops and multiedges),
\begin{equation}\label{hom-reqn-in-intro}
\hom(G, H) = \hom(G, H'),
\end{equation}
then the graphs $H$ and $H'$ are isomorphic as weighted graphs, i.e.,
there is a bijective map from $H$ to $H'$ that
preserves all vertex and edge weights. 
\end{theorem}

Theorem~\ref{thm:intro-gh-both-twin-free-k=0} is the special
case of $k=0$ of the more general Theorem~\ref{thm:gh-main} which deals with
$k$-labeled graphs. In Section~\ref{sec:Implications}
we also determine the dimensions of the 
corresponding graph algebras in terms of
the rank of the so-called connection tensors, 
introduced in~\cite{Cai-Govorov-2019}. These  
improve the corresponding theorems in~\cite{Lovasz-2006,Schrijver-2009,Regts-2013}
as follows.

From the main theorem (Theorem 2.2) of~\cite{Lovasz-2006}
we generalize from positive vertex weights and real edge weights
to arbitrary weights.
The main technique in~\cite{Lovasz-2006} is algebraic.
The proof relies on the notion of quantum graphs and structures built from them,
and uses idempotent elements in the graph algebras.
Similarly, from the isomorphism theorem in~\cite{Schrijver-2009} 
we generalize from unit vertex weights and complex edge weights to arbitrary weights.
The results in~\cite{Schrijver-2013} also deal with complex edge weights 
but unit vertex weights.
Our results allow directed and undirected weighted graphs $H$. 
Theorem~\ref{thm:gh-main} also
weakens the condition (\ref{hom-reqn-in-intro})
on $G$ to simple graphs (i.e., no multiedges or loops).
Schrijver's proof technique~\cite{Schrijver-2009}
 is different from that of Lov\'asz~\cite{Lovasz-2006},
but is also algebraic and built on quantum graphs.
He uses a Reynolds operator and the M\"obius transform (of a graph).
Also in~\cite{Schrijver-2009,Schrijver-2013} the graphs $G$ 
may have multi-loops
and multi-edges.
The results of Lov\'asz~\cite{Lovasz-2006} and Schrijver~\cite{Schrijver-2009}
are incomparable.
While requiring all vertex weights to be positive
is not unreasonable,
it is nonetheless a severe restriction, and has been a technical obstacle
to previous work.
In Regts' thesis~\cite{Regts-2013}, multiple theorems were proved with
the explicit requirement that all (nonempty) sums of vertex 
weights be nonzero, which circumvented this issue.
In this paper, we allow arbitrary
vertex weights with no assumptions.
In particular, $H$ can have arbitrary complex vertex and edge weights.

However, more than the explicit strengthening of the theorems,
we believe the most 
noteworthy aspect of this work is that
we found a direct elementary argument that
bypassed various technical obstacles
and unified all previously known versions.
We can also show that the only restriction---$\bbF$ has
 characteristic $0$---cannot be removed, and thus
 our results are the most general extensions on graphs.
We give counterexamples for fields of finite characteristic 
in Section~\ref{sec:counterexample}.

This line of work has already led to significant applications
in the graph limit literature, such as on quasi-random graphs~\cite{Lovasz-Sos-2008}.
In~\cite{Lovasz-Szegedy-2009} Lov\'asz and B.~Szegedy also studied
these graph algebras  where ``contractors'' and ``connectors'' 
are used.
In our treatment these ``contractors'' and ``connectors'' can also be constructed
with simple graphs.
 
In terms of applications to complexity theory,
there has been a series of significant complexity dichotomy
theorems on counting graph homomorphisms which
show that the function $\hom(\cdot, H)$
is either P-time computable or \#P-hard, \emph{depending} on $H$~\cite{Dyer-Greenhill-2000,
Dyer-Greenhill-corrig-2004,
Bulatov-Grohe-2005,Grohe-Thurley-2011,Thurley-2009,
Goldberg-et-al-2010,
Cai-Chen-Lu-2013,
Dyer-Goldberg-Paterson-2007,
Cai-Chen-2019}.
These theorems differ in the scope of what types of $H$ are allowed,
from 0-1 valued to complex valued, from undirected to directed.
In all these theorems 
 a P-time tractability condition on $H$ is given
such that if $H$  satisfies the condition then
$\hom(\cdot, H)$ is P-time computable, otherwise
$\hom(\cdot, H)$ is \#P-hard.  In the latter case,
the theorem asserts that
there is a P-time reduction
from a canonical \#P-hard problem to $\hom(\cdot, H)$.
\emph{However}, the proof of these dichotomy
 theorems requires  various \emph{pinning lemmas}, and some are
proved  nonconstructively.
For example,
for undirected complex weighted graphs~\cite{Cai-Chen-Lu-2013} 
it was unknown how to make one pinning lemma constructive.
Consequently, there was no known algorithm 
 to produce a \#P-hardness reduction from $H$.
Because
the proof in this paper is
constructive, it can be applied 
to make all these dichotomy theorems
\emph{effective} in the following sense:  we can obtain an algorithm
that for any $H$ either outputs a P-time algorithm 
solving $\hom(\cdot, H)$
or a P-time reduction from a canonical \#P-hard problem
to  $\hom(\cdot, H)$.

Theorem~\ref{thm:intro-gh-both-twin-free-k=0}
weakens the requirement of $G$ to \emph{simple graphs}
in order that (\ref{hom-reqn-in-intro}) 
 implies isomorphism of $H$ and $H'$.
In~\cite{DGR-2018}, Grohe, Dell, and Rattan
studied the restriction of $G$ 
to trees or graphs of bounded tree width in (\ref{hom-reqn-in-intro}),
and relate this to the Weisfeiler-Leman algorithm 
for graph isomorphism. They showed that
 (\ref{hom-reqn-in-intro}) for all trees $G$ is equivalent to
indistinguishability by
the color refinement algorithm (1-dimensional W-L algorithm).
Corresponding results were also obtained for 
graphs of tree width $k$, and the $k$-dimensional W-L algorithm.
It is known that $k$-dimensional W-L algorithm is not
sufficient for graph isomorphism~\cite{CFI92}.
It is an interesting open problem whether  (\ref{hom-reqn-in-intro})
 can be further
weakened from simple graphs, and to characterize the equivalence relation
defined by (\ref{hom-reqn-in-intro}) for various
restricted classes of graphs.

\section{Preliminaries}\label{sec:preliminaries}
We first recap the notion of weighted graph 
homomorphisms~\cite{Freedman-Lovasz-Schrijver-2007}, 
but state it for an arbitrary field $\bbF$.
We let $[k] = \{ 1, \ldots, k \}$ for integer $k \ge 0$.
In particular, $[0] = \emptyset$.
By convention $\bbF^0  = \{\emptyset\}$,
and $0^0 = 1$ in $\bbZ$, $\bbF$, etc.
Often we discuss both directed and undirected graphs together.

An ($\bbF$-)\textit{weighted graph} $H$ is a finite (di)graph with a 
weight $\alpha_H(i) \in \bbF \setminus \{0\}$ associated with 
each vertex $i$ ($0$-weighted vertices can be deleted) 
and a weight $\beta_H(i, j) \in \bbF$ associated with 
each edge $ij$ (or loop if $i=j$). For undirected graphs,
$\beta_H(i, j)= \beta_H(j, i)$.
It is convenient to assume that
$H$ is a complete graph with a loop at all nodes by adding
all missing edges and loops with weight $0$.
Then $H$ is described by an integer $q = |V(H)| \ge 0$ ($H$ can be the empty graph), a nowhere zero vector $\alpha = (\alpha_H(1), \ldots, \alpha_H(q)) \in \bbF^q$
and a matrix $B = (\beta_H(i, j)) \in \bbF^{q \times q}$.
An isomorphism from $H$ to $H'$ is a bijective map
from $V(H)$ to $V(H')$ that preserves vertex and edge weights.

According to~\cite{Freedman-Lovasz-Schrijver-2007},
let $G$ be an unweighted graph (with possible 
multiple edges, but no loops) and $H$ a weighted graph given by $(\alpha, B)$,
we define
\begin{equation}\label{eqn:hom-def}
\hom(G, H) = \sum_{\phi \colon V(G) \to V(H)}
\alpha_\phi \hom_\phi(G, H)
\end{equation}
where
\begin{equation}\label{eqn:hom-def-vertex-hom-subphi}
\alpha_\phi
=  \prod_{u \in V(G)} \alpha_H(\phi(u)),~~~~
\hom_\phi(G, H) = 
 \prod_{u v \in E(G)} \beta_H(\phi(u), \phi(v)).
\end{equation}
The unweighted case is when all node weights are $1$ and all edge weights
are $0$-$1$ in $H$,
and $\hom(G, H)$ is the number of homomorphisms from $G$ into $H$.

A $k$-labeled graph ($k \ge 0$) is a finite graph in which $k$ nodes are labeled by $1, 2, \ldots, k$ (the graph can have any number of unlabeled nodes). Two $k$-labeled graphs are isomorphic if there is a label-preserving isomorphism between them. 
$U_k$ denotes the $k$-labeled graph on $k$ nodes with no edges. In particular, 
$U_0$ is the empty graph with no nodes and no edges.
The  \textit{product} of  two $k$-labeled graphs $G_1$ and $G_2$
 is defined as follows: take their disjoint union, and then identify nodes with the same label. Hence for two $0$-labeled graphs, $G_1 G_2 = G_1 \sqcup G_2$ (disjoint union). Clearly, the graph product is associative and commutative with the identity $U_k$, so the set of all (isomorphism classes) of $k$-labeled graphs together with the product operation forms a commutative monoid which we denote by $\calPLG[k]$.
We denote by $\calPLG^{\simp}[k]$ the submonoid of
simple graphs in $\calPLG[k]$; these are graphs with no loops, at most one
edge between any two vertices $i$ and $j$, and no edge between labeled vertices.
A directed labeled graph is simple if its underlying undirected one
is simple; 
in particular, for any $i$ and $j$, we require that
if $i \rightarrow j$ is an edge, then $j \rightarrow i$ is not an edge.
Clearly,  $\calPLG^{\simp}[k]$  is closed under the product operation
(for both directed and undirected types).

Fix a weighted graph $H = (\alpha, B)$.
For any $k$-labeled graph $G$ and mapping $\psi \colon [k] \to V(H)$, let
\begin{equation}\label{eqn:partial-hom-dec}
\hom_\psi(G, H) = \sum_{\text{\small $\substack{\phi \colon V(G) \to V(H) \\ \phi \text{ extends } \psi}$}} \frac{\alpha_\phi}{\alpha_\psi} \hom_\phi(G, H), 
\end{equation}
where $\phi  \text{ extends } \psi$ means that
 if $u_i \in V(G)$ is labeled by $i \in [k]$
then $\phi(u_i) = \psi(i)$, and
 $\alpha_\psi= \prod_{i=1}^k \alpha_H(\psi(i))$, $\alpha_\phi
= \prod_{v \in V(G)} \alpha_H(\phi(v))$, so $\frac{\alpha_\phi}{\alpha_\psi}$
is the product of vertex weights of $\alpha_\phi$ \emph{not} in $\alpha_\psi$. Then
\begin{equation}\label{eqn:full-hom-dec}
\hom(G, H) = \sum_{\psi \colon [k] \to V(H)} \alpha_\psi \hom_\psi(G, H).
\end{equation}
When $k = 0$, we only have the empty map $\emptyset$ with the domain $\emptyset$.
Then $\hom(G, H) = \hom_\emptyset(G, H)$ for every $G \in \calPLG[k]$.
The functions $\hom_\psi(\cdot, H)$, where $\psi \colon [k] \to V(H)$ and $k \ge 0$, satisfy
\begin{equation}\label{eq:part-hom-multiplicativity}
\begin{lcases}
&\hom_\psi(G_1 G_2, H) = \hom_\psi(G_1, H) \hom_\psi(G_2, H), \quad G_1, G_2 \in \calPLG[k], \\
&\hom_\psi(U_k, H) = 1.
\end{lcases}
\end{equation}

Given a directed or undirected $\bbF$-weighted graph $H$,
we call two vertices $i, j \in V(H)$ \textit{twins} if for every vertex $\ell \in V(H)$,
$\beta_H(i, \ell) = \beta_H(j, \ell)$ and $\beta_H(\ell, i) = \beta_H(\ell, j)$.
Note that in this definition the universal quantifier over $\ell$ includes
$\ell = i$ and $\ell=j$. Also note that the vertex weights $\alpha_H(w)$ do not participate in this definition.
If $H$ has no twins, we call it twin-free.

The twin relation 
partitions $V(H)$ into nonempty equivalence classes,
$I_1, \ldots, I_s$ where $s \ge 0$.
We can define a  \textit{twin contraction} graph  $\widetilde H$,
having $I_1, \ldots, I_s$  as vertices,
with  vertex weight $\sum_{t \in I_r} \alpha_H(t)$ for $I_r$,
and edge  weight from $I_r$ to $I_q$ to be $\beta_H(u, v)$
for some arbitrary $u \in I_r$ and $v \in I_q$.
After that, we remove all vertices in  $\widetilde H$ 
with zero vertex weights together with all incident edges (still called 
$\widetilde H$).
This defines a twin-free $\widetilde H$.
Clearly, $\hom(G, H) = \hom(G, \widetilde H)$ for all $G$.

We denote by $\Isom(H, H')$ the set of $\bbF$-weighted graph isomorphisms from $H$ to $H'$
and by $\Aut(H)$ the group of ($\bbF$-weighted) graph automorphisms of $H$.

It is obvious that for directed (or undirected) $\bbF$-weighted graphs $H$ and $H'$,
and the maps $\varphi \colon [k] \to V(H)$ and $\psi \colon [k] \to V(H')$ such that $\psi = \sigma \circ \varphi$
for some isomorphism $\sigma \colon V(H) \to V(H')$ from $H$ to $H'$,
we have $\hom_\varphi(G, H) = \hom_\psi(G, H')$ for every $G \in \calPLG[k]$.

\section{Our results}
Theorem~\ref{thm:intro-gh-both-twin-free-k=0} is a direct consequence of 
the 
case $k=0$ of the following Main Theorem.
\begin{theorem}\label{thm:gh-main}
Let $\bbF$ be a field of characteristic $0$.
Let $H, H'$ be (directed or undirected) $\bbF$-weighted graphs
such that $H$ is twin-free and $m = |V(H)| \ge m' = |V(H')|$.
Suppose $\varphi \colon [k] \to V(H)$ and $\psi \colon [k] \to V(H')$ where $k \ge 0$.
If $\hom_\varphi(G, H) = \hom_\psi(G, H')$ for every $G \in \calPLG^{\simp}[k]$,
then there exists an isomorphism of $\bbF$-weighted graphs
$\sigma \colon V(H) \to V(H')$ from $H$ to $H'$
such that $\psi = \sigma \circ \varphi$
(a fortiori, $H'$ is twin-free and $m = m'$).
\end{theorem}

In Section~\ref{sec:Implications} we will give our results
about the space of such isomorphisms, expressed in terms of
the dimensions of the corresponding graph algebras.
However, all results stated in this section will be proved in Sections~\ref{sec:MainProof} and \ref{sec:MainProof-dir}.

In Theorem~\ref{thm:gh-main} we require  ${\rm char}~\bbF = 0$, and
this is also assumed in the following 
Corollaries~\ref{cor:gh-both-twin-free-k=0}
to~\ref{cor:gh-only-ew-k=0-matrix-version}.
In  Section~\ref{sec:counterexample} we show that this
condition is necessary.
The following two corollaries extend
Lov\'asz's theorem (and lemmas that are of independent interest)
in~\cite{Lovasz-2006} from
real edge weight and positive vertex weight.
Furthermore,  they hold for both directed and undirected graphs,
and the condition on $G$
 is weakened so that it is sufficient to assume it for simple graphs only.
The fact that the theorem holds under the condition 
$\hom_\varphi(G, H) = \hom_\psi(G, H')$ for loopless labeled graphs
 $G$ is important in making the complexity dichotomies \emph{effective} in 
the sense defined in Section~\ref{sec:intro}.

%
%

\begin{corollary}\label{cor:gh-both-twin-free-k=0}
Let $H, H'$ be  $\bbF$-weighted twin-free graphs,
either both directed or both undirected.
If $\hom(G, H) = \hom(G, H')$ for every \emph{simple} graph $G$,
then 
$H$ and $H'$ are isomorphic as $\bbF$-weighted graphs.
\end{corollary}

For edge weighted graphs with unit vertex weight,
the requirement of twin-freeness can be dropped.
The following two corollaries directly generalize
Schrijver's theorem (Theorem~2) in~\cite{Schrijver-2009}.
Corollary~\ref{cor:gh-only-ew-k=0-matrix-version}
 is a restatement of Corollary~\ref{cor:gh-only-ew-k=0}
using the terminology in~\cite{Schrijver-2009}.
Here we strengthen his theorem by requiring the 
condition  $\hom(G, H) = \hom(G, H')$ for only simple graphs $G$.
Also our result holds 
for fields $\bbF$ of characteristic $0$ generalizing from $\mathbb{C}$, and for directed as well as undirected graphs.
\begin{corollary}\label{cor:gh-only-ew-k=0}
Let $H, H'$ be (directed or undirected) $\bbF$-edge-weighted graphs.
If $\hom(G, H) = \hom(G, H')$ for every \emph{simple}
graph $G$,
then
$H$ and $H'$ are isomorphic as $\bbF$-weighted graphs.
\end{corollary}

\begin{corollary}\label{cor:gh-only-ew-k=0-matrix-version}
Let $A \in \bbF^{m \times m}$ and $A' \in \bbF^{m' \times m'}$.
Then $\hom(G, A) = \hom(G, A')$ for 
every \emph{simple} graph $G$
iff $m = m'$ and there is a permutation matrix $P \in \bbF^{m \times m}$ such that $A' = P^T A P$.
\end{corollary}

Our proof of Theorem~\ref{thm:gh-main} will show that
for any given $H, H'$, there is an explicitly constructed finite family
of graphs in $\calPLG^{\simp}[k]$ such that
the condition ``for all  $G \in \calPLG^{\simp}[k]$''
can be replaced with ``for all $G$ in this family'',
and thus we have an explicit finitary family of test graphs.
Moreover, this provides an explicit set of ``witnesses'' that can be
used to make various complexity dichotomy theorems
for graph homomorphism \emph{effective},
in particular, making the pinning steps in~\cite{Cai-Chen-Lu-2013}
computable, which was an open problem.

\section{The Vandermonde Argument}


We start with an exceedingly simple lemma, based on which all of our
results will be derived. We will call this lemma and its corollary 
the \emph{Vandermonde Argument}.
\begin{lemma}\label{lem:simpleVandermonde}
Let $n \ge 0$, and $a_i, x_i \in \bbF$ for $1 \le i \le n$.
Suppose 
\begin{equation}\label{eq:simple-Vandermonde-condition}
\sum_{i = 1}^n a_i x_i^j = 0, ~~~~\text{for all}~~ 0 \le j < n.
\end{equation}
Then for any function $f \colon \bbF \to \bbF$, we have
$\sum_{i = 1}^n a_i f(x_i) = 0$.
If (\ref{eq:simple-Vandermonde-condition}) is true
for $1 \le j \le n$, then the same conclusion holds
for any  function $f$ satisfying $f(0) =0$.
\end{lemma}

%
\begin{proof}
The statement is vacuously true if $n=0$, since an empty
sum is 0.
We may assume $n \ge 1$.
We partition $[n]$ into $\bigsqcup_{\ell=1}^{p} I_\ell$ 
such that $i, i'$ belong to
the same $I_\ell$ iff $x_i = x_{i'}$. Then (\ref{eq:simple-Vandermonde-condition}) 
is a Vandermonde system of rank $p$ with a solution
$(\sum_{i \in I_\ell} a_i)_{\ell \in [p]}$.
Thus $\sum_{i \in I_\ell} a_i =0$ for all $1 \le \ell \le p$.
It follows that $\sum_{i = 1}^n a_i f(x_i) = 0$
for any function $f \colon \bbF \to \bbF$.
If  (\ref{eq:simple-Vandermonde-condition})
is true for  $1 \le j \le n$, then the same proof works except
when some $x_i=0$. In that case, we can separate out 
the term $\sum_{i \in I_{\ell_0}} a_i$ for the unique $I_{\ell_0}$ that contains
this $i$, and we get a Vandermonde system of rank $p-1$ on the other terms
$(\sum_{i \in I_\ell} a_i)_{\ell \in [p], \ell \ne \ell_0}$, which must be all zero.
\end{proof}

By iteratively applying Lemma~\ref{lem:simpleVandermonde},
we get the following Corollary.

\begin{corollary}\label{cor:vand-cancellation-gen}
Let $I$ be a finite (index) set, $s \ge 1$,
 and $a_i, b_{i j}  \in \bbF$ for all $i \in I, j \in [s]$.
Further, let $I = \bigsqcup_{\ell \in [p]} I_\ell$
be the partition of $I$
into equivalence classes, where
$i, i'$ are equivalent iff $b_{i j} = b_{i' j}$ for all $j \in [s]$.
If $\sum_{i \in I} a_i \prod_{j \in [s]} b_{i j}^{\ell_j} = 0$,
for all choices of $(\ell_1, \ldots, \ell_s)$ where each $0 \le \ell_j < |I|$,
then $\sum_{i \in I_\ell} a_i = 0$ for every $\ell \in [p]$.
\end{corollary}
\begin{proof}
First, we define an equivalence relation
where $i, i'$ belong to the same equivalence
class $\widetilde{I}$ iff $b_{i s} = b_{i' s}$. 
For any $\widetilde{I}$,
choose  $f$ with $f(x)=1$ for $x= b_{i s}$ where $i \in
\widetilde{I}$, and $f(x)=0$ otherwise.
By  Lemma~\ref{lem:simpleVandermonde}, 
we get
$\sum_{i \in \widetilde{I}} a_i \prod_{j \in [s-1]} b_{i j}^{\ell_j} = 0$,
for an arbitrary $\widetilde{I}$, 
and all $0 \le \ell_j < |I|, j \in [s-1]$.
Corollary~\ref{cor:vand-cancellation-gen}
 follows after applying  Lemma~\ref{lem:simpleVandermonde} $s$ times.
\end{proof}

\section{Undirected graphs}\label{sec:MainProof}



To illustrate the simplicity of the main proof idea,
in this section we first present a proof of 
Theorem~\ref{thm:gh-main} for undirected graphs.

We may assume that $H$ is on the vertex set $V(H) = [m]$,
given by vertex and edge weights $(\alpha_i)_{i \in [m]} \in (\bbF \setminus \{0\})^m$,
$(\beta_{i j})_{i, j \in [m]} \in \bbF^{m \times m}$.
Similarly, $H'$ is on $V(H') = [m']$, given by
$(\alpha'_i)_{i \in [m']} \in (\bbF \setminus \{0\})^{m'} $ and $(\beta'_{i j})_{i, j \in [m']} \in \bbF^{m' \times m'}$.

We first make a technical condition of
``super surjectivity''; it will be removed later.
\begin{lemma}\label{lem:gh-highly-surj}
Let $H, H'$ be undirected $\bbF$-weighted graphs
such that $H$ is twin-free and
$m \ge m'$.
Suppose 
 $\varphi \colon [k] \to V(H)$ and $\psi \colon [k] \to V(H')$ where $k \ge 0$.
Assume $\varphi$ is ``super surjective'', namely:
 $|\varphi^{-1}(u)| \ge 2m^2$ for every $u \in V(H)$.
If $\hom_\varphi(G, H) = \hom_\psi(G, H')$ for every $G \in \calPLG^{\simp}[k]$, 
then there exists an isomorphism $\sigma$ from $H$ to $H'$
such that $\psi = \sigma \circ \varphi$.
\end{lemma}
\begin{proof}
Assume $m \ge 1$ (the case  $m = 0$ is trivial). Taking any $u \in V(H)$, we get
$k \ge |\varphi^{-1}(u)| \ge 2 m^2 >0$ and thus $m' \ge |\psi([k])| \ge 1$.
For each $\kappa = (b_i)_{i \in [k]} \in \{0, 1\}^k$,
we 
define a  graph $G_\kappa \in \calPLG^{\simp}[k]$:
\begin{quote}
$V(G_\kappa) = \{u_1, \ldots, u_k, v\}$,
with each $u_i$ labeled $i$.
For each $i \in [k]$, there is an  edge 
$(v, u_i)$ iff $b_i = 1$. There are no other edges.
\end{quote}

We now define a specific set of $G_\kappa$, given $\varphi$ and $\psi$.
We can partition $[k] = \bigsqcup_{i = 1}^m I_i$
where each $I_i = \varphi^{-1}(i)$ and
$|I_i| \ge 2m^2$.
For every $i \in [m]$, since $|I_i| \ge 2m m'$,
there exists  
$J_i \subseteq I_i$ such that $|J_i| \ge 2 m >0$
and the restriction $\psi_{|J_i}$ takes a constant value $s(i)$,
for some function $s \colon [m] \to [m']$.
Next, for each $i \in [m]$ and
for every  $0 \le k_i < 2m$,
we can fix  $K_i\subset J_i$,
with $|K_i| = k_i$.
Then
we let the tuple $\chi = \chi(K_1, \ldots, K_m)
\in \{0,1\}^k$
take $\chi_{|K_i} = 1$,  $i \in [m]$,  and all other entries are $0$.
Let $R$ be the set of all such tuples $\chi$ for every possible choice of 
$(k_i)_{i \in [m]}$ with $0 \le k_i < 2m$.

Then $\hom_\varphi(G_\chi, H) = \hom_\psi(G_\chi, H')$
for every $G_\chi$ with $\chi \in R$ is expressed by: For all
$0 \le k_j < 2m$, $j \in [m]$,
\begin{equation}\label{eq:hom-G-kappa-1st}
\sum_{i = 1}^m \alpha_i \prod_{j = 1}^m \beta_{i j}^{k_j}  = 
\sum_{i = 1}^{m'} \alpha'_i \prod_{j = 1}^m (\beta'_{i s(j)})^{k_j}.
\end{equation}
In (\ref{eq:hom-G-kappa-1st}) the sums on $i$ 
come from  assigning  $v \in V(G_\chi)$
 to $i \in V(H)$ or to $i \in V(H')$, respectively.

Because $H$ is twin-free the $m$-tuples $(\beta_{i j})_{j \in [m]} \in \bbF^{m}$
for $i \in [m]$ are pairwise distinct.
In (\ref{eq:hom-G-kappa-1st}) the sum in the LHS 
has $m$ terms, while the sum in the RHS has  $m' \le m$ terms.
Transferring the RHS to the LHS we get
at most $2m$ terms.
Now we apply Corollary~\ref{cor:vand-cancellation-gen} to the sum
obtained by moving all terms of the RHS to the LHS in (\ref{eq:hom-G-kappa-1st}).
By the pairwise distinctness of
the $m$-tuples $(\beta_{i j})_{j \in [m]}
\in \bbF^{m}$ for $i \in [m]$, 
and since there are only $m' \le m$ terms from the RHS and every $\alpha_i \ne 0$,
we see that each term from the LHS of
(\ref{eq:hom-G-kappa-1st}) must be canceled by exactly one term from the RHS. 
And this can only occur if $m = m'$,
and there is a bijective map $\sigma \colon [m] \to [m]$:
\begin{equation}\label{eq:alpha-beta-cancellation-itentities}
\alpha_i = \alpha'_{\sigma(i)} \quad \text{for } i \in [m],
\quad  \quad (\beta_{i j})_{j \in [m]} = (\beta'_{\sigma(i) s(j)})_{j \in [m]} \quad \text{for } i \in [m].
\end{equation} 
Since $H, H'$ are undirected graphs,
we also have $\beta_{i j} = \beta_{ji} = \beta'_{ \sigma(j) s(i)}
=  \beta'_{s(i) \sigma(j)}$ for $i, j \in [m]$.

Next we show that $s$ is bijective. If for some $x, y \in [m]$ we have $s(x) = s(y)$,
then
\[
(\beta_{x j})_{j \in [m]} 
= (\beta'_{s(x) \sigma(j)})_{j \in [m]} 
= (\beta'_{s(y) \sigma(j)})_{j \in [m]} 
= (\beta_{y j})_{j \in [m]}.
\]
Since $H$ is twin-free, we conclude that $x = y$. Thus 
the map $s \colon [m] \to [m]$ 
is injective and so (since $[m]$ is finite) $s$ is bijective.
However, $\sigma \colon [m] \to [m]$ is also bijective, 
so it follows from (\ref{eq:alpha-beta-cancellation-itentities})
that the tuples $(\beta'_{i j})_{j \in [m]}$ for $i \in [m]$
are pairwise distinct, which means that $H'$ is twin-free as well.

Next we show $\psi_{|I_i} = s(i)$ for all $i \in [m]$.
If for all   $i \in [m]$, we have  $J_i = I_i$, then we are done.
Otherwise, take any $w \in [m]$ such that $J_w$ is a proper subset of
$I_w$ and 
we take any $t \in I_w \setminus J_w$.
Observe that $t \notin K_i$ for all $i \in [m]$.
In particular, $\chi(t) = 0$ for each $\chi \in R$.

For each $\chi \in R$,
let $\chi_+$ be the tuple obtained from $\chi$ by reassigning $\chi(t)$
(changing its $t$-th entry) from $0$ to $1$
and let $R_{+}$ be the set of all such $\chi_+$.

Then $\hom_\varphi(G_{\kappa}, H) = \hom_\psi(G_{\kappa}, H')$
for every $G_{\kappa}$ with $\kappa \in R_{+}$ is expressed as
(recall that we have already proved that $m'=m$)
\[
\sum_{i = 1}^m \alpha_i \beta_{i w} \prod_{j = 1}^m \beta_{i j}^{k_j}  = \sum_{i = 1}^m \alpha'_i \beta'_{i \psi(t)} \prod_{j = 1}^m  (\beta'_{i s(j)})^{k_j},
\]
which can be compared to (\ref{eq:hom-G-kappa-1st}), and here
for
$\kappa \in R_{+}$
we have one extra edge $(v,u_t)$ in $G_{\kappa}$,
and $\varphi(t) =w$ since $t \in I_w$.
So this holds for every $0 \le k_j < 2m$ where $j \in [m]$.
Transferring the RHS to the LHS and using (\ref{eq:alpha-beta-cancellation-itentities}), we get
\[
\sum_{i = 1}^m (\alpha_i \beta_{i w} - \alpha'_{\sigma(i)} \beta'_{\sigma(i) \psi(t)}) \prod_{j = 1}^m \beta_{i j}^{k_j}  = 0
\]
for every $0 \le k_j < 2m$ where $j \in [m]$.
Since $\alpha_i = \alpha'_{\sigma(i)} \not = 0$,
and the tuples $(\beta_{i j})_{j \in [m]}$ for $i \in [m]$ are pairwise distinct,
by Corollary~\ref{cor:vand-cancellation-gen},
we get $\beta_{i w} -  \beta'_{\sigma(i) \psi(t)} = 0$ for $i \in [m]$.
On the other hand, by~(\ref{eq:alpha-beta-cancellation-itentities}),
$\beta_{i w} = \beta'_{\sigma(i) s(w)}$  for $i \in [m]$.
It follows that $\beta'_{\sigma(i) \psi(t)} = \beta'_{\sigma(i) s(w)}$
for $i \in [m]$.
However, since $\sigma \colon[m] \to [m]$ is a bijection and,
as shown before, $H'$ is twin-free, this implies that  $\psi(t) = s(w)$. 
Recall that  $\psi_{|J_w} = s(w)$.
This proves that on $I_w \setminus J_w$,  $\psi$ also takes
the constant value $s(w)$.  
Thus $\psi_{|I_i} = s(i)$ for all $i \in [m]$.

Next we prove that $\beta_{i j} = \beta'_{\sigma(i) \sigma(j)}$ for
all $i, j \in [m]$,
i.e., $\sigma$ preserves the edge weights.

For each $\lambda = (b_i)_{i \in [k]}$,
$\tau = (c_i)_{i \in [k]} \in \{0, 1\}^k$,
we define a graph $G_{\lambda, \tau}  \in \calPLG^{\simp}[k]$:
\begin{quote}
$V(G_{\lambda, \tau}) = \{u_1, \ldots, u_k, v, v'\}$,
with each $u_i$ labeled $i$.
There is an edge  $(v, v')$ and,
for each $i \in [k]$, there is an edge
$(v,u_i)$ if $b_i = 1$,
and  an edge $(v', u_i)$ if $c_i = 1$.
There are no other edges. 
\end{quote}

Let $R^2 = R \times R$.
By the definition of $R$, every $(\lambda, \tau) \in R^2$
corresponds to a sequence of pairs $(K_i, L_i)$, $i \in [m]$,
where $K_i \subset J_i
\subseteq I_i$ and $L_i \subset J_i
\subseteq I_i$
such that $|K_i| = k_i, |L_i| = \ell_i$,
where $0 \le k_i, \ell_i < 2m$,
 and 
$\lambda_{|K_i}  = 1$,  
$\tau_{|L_i} = 1$, for $i \in [m]$,
 and all other entries are $0$.
Moreover, the above correspondence between $(\lambda, \tau) \in R^2$
and the sequences of pairs $(K_i, L_i)$, $i \in [m]$, is bijective.
In particular, every possible choice of $0 \le k_i, \ell_i < 2m$, $i \in [m]$, is accounted for exactly once under this correspondence.

Then $\hom_\varphi(G_{\lambda, \tau}, H) = \hom_\psi(G_{\lambda, \tau}, H')$
for every $G_{\lambda, \tau}$ with $(\lambda, \tau) \in R^2$  is
expressed as
\[
\sum_{i, j = 1}^m \alpha_{i} \alpha_{j} \beta_{i j} \prod_{r = 1}^m 
(\beta_{i r}^{k_r} \beta_{jr}^{\ell_r}) =
\sum_{i, j = 1}^m \alpha'_{i} \alpha'_{j} \beta'_{i j} \prod_{r = 1}^m 
\left( (\beta'_{i s(r)})^{k_r}  (\beta'_{j s(r)})^{\ell_r} \right).
\]
This holds for all $0 \le k_r, \ell_r < 2m$ 
where $r \in [m]$.
Transferring the RHS to the LHS and using (\ref{eq:alpha-beta-cancellation-itentities}), we get
\[
\sum_{i, j = 1}^m (\alpha_{i} \alpha_{j} \beta_{i j} - \alpha_{i} \alpha_{j} \beta'_{\sigma(i) \sigma(j)}) \prod_{r = 1}^m (\beta_{i r}^{k_r} \beta_{jr}^{\ell_r} ) = 0
\]
for every $0 \le k_r, \ell_r < 2m$ where $r \in [m]$.
Since the tuples 
$(\beta_{i r},  \beta_{j r})_{r \in [m]}$ for $i, j \in [m]$ are 
pairwise distinct,
by Corollary~\ref{cor:vand-cancellation-gen}, $\alpha_{i} \alpha_{j} \beta_{i j} - \alpha_{i} \alpha_{j} \beta'_{\sigma(i) \sigma(j)} = 0$ for $i, j \in [m]$.
But all $\alpha_i \ne 0$, so
\begin{equation}\label{eq:beta-preservation}
\beta_{i j} = \beta'_{\sigma(i) \sigma(j)} \quad \text{for } i, j \in [m].
\end{equation}
This 
means that the bijection $\sigma \colon [m] \to [m]$ preserves the edge weights in addition to the vertex weights by~(\ref{eq:alpha-beta-cancellation-itentities}).
Hence $\sigma \colon [m] \to [m]$ is an isomorphism of $\bbF$-weighted graphs from $H$ to $H'$.

Finally, we show that $\psi = \sigma \circ \varphi$.
From~(\ref{eq:alpha-beta-cancellation-itentities}) and~(\ref{eq:beta-preservation}), we have
$\beta'_{\sigma(i) s(j)} = \beta_{i j} = \beta'_{\sigma(i) \sigma(j)}$ 
for $i, j \in [m]$.
As $H'$ is twin-free and $\sigma$ is bijective we get $\sigma(j) = s(j)$ for $j \in [m]$.
Now let $x \in [k]$. Then $x \in I_i$ for some $i \in [m]$.
Therefore $\varphi(x) = i$ and so $\psi(x) = s(i) = \sigma(i) = \sigma(\varphi(x))$ confirming $\psi  = \sigma \circ \varphi$.
\end{proof}

\begin{proof}[Proof of Theorem~\ref{thm:gh-main} for undirected graphs]
Consider an
arbitrary $\ell \ge k$ and let $G \in \calPLG^{\simp}[\ell]$ be any
$\ell$-labeled graph.
Let $G^* = \pi_{[k]}(G)$ be the graph obtained by unlabeling the labels not in $[k]$ from $G$ (if $k = \ell$, then $G^* = G$).
Clearly, $G^* \in \calPLG^{\simp}[k]$ so $\hom_\varphi(G^*, H) = \hom_\psi(G^*, H')$.
Expanding the sums on the LHS and the RHS of this equality
representing the  
maps $\varphi, \psi$
along the vertices formerly labeled by $[\ell] \setminus [k]$,
and then regrouping the terms corresponding to the same extension maps
 of $\varphi$ from $[\ell]$ to $[m]$ and of $\psi$ from $[\ell]$ to $[m']$,
 respectively,
and then bringing back the labels from $[\ell] \setminus [k]$, we get
\begin{equation}\label{eq:gh-expanded-along-labels}
\sum_{\substack{\mu \colon [\ell] \to [m] \\ \mu \text{ extends } \varphi }} \Bigg(\prod_{k <i \le\ell} \alpha_{\mu(i)}\Bigg) \hom_\mu(G, H) =
\sum_{\substack{\nu \colon [\ell] \to [m'] \\ \nu \text{ extends } \psi}} \Bigg(\prod_{k <i \le\ell} \alpha'_{\nu(i)}\Bigg) \hom_\nu(G, H')
\end{equation}
for every $G \in \calPLG^{\simp}[\ell]$.
(Here if $\ell =k$, the empty product $\prod_{k <i \le\ell}$ is 1.)

Now choose $\ell \ge k$ so that
we can extend $\varphi$ to a map $\eta \colon [\ell] \to [m]$ 
such that $|\eta^{-1}(u)| \ge 2m^2$ for every $u \in V(H)$. 
Clearly $\ell \le k + 2m^3$ suffices. 
(If $\varphi$ already satisfies the property, we can
take $\ell = k$ and $\eta = \varphi$.)
We fix $\ell$ and $\eta$ to be such.
Define
\begin{align*}
I &= \{ \mu \colon [\ell] \to [m] \mid \big(\mu_{|[k]} = \varphi \big) \wedge 
\big((\exists \sigma \in \Aut(H)) \; \mu = \sigma \circ \eta \big) \}, \\
J &= \{ \nu \colon [\ell] \to [m'] \mid \big(\nu_{|[k]} = \psi \big) \wedge 
\big((\exists \sigma \in \Isom(H, H')) \; \nu = \sigma \circ \eta \big) \}.
\end{align*}
Obviously, $\eta \in I$ so $I \ne \emptyset$.
For now, we do not exclude the possibility $J = \emptyset$
but our goal is to show that this is not the case.
If $\mu \colon [\ell] \to V(H)$ extends $\varphi$ but $\mu \notin I$, then by Lemma~\ref{lem:gh-highly-surj},
there exists a graph $G_{\eta, \mu} \in \calPLG^{\simp}[\ell]$
such that $\hom_\eta(G_{\eta, \mu}, H) \ne \hom_\mu(G_{\eta, \mu}, H)$.
Similarly, if $\nu \colon [\ell] \to V(H')$ extends $\psi$ but $\nu \notin J$, then
 by  Lemma~\ref{lem:gh-highly-surj},
there exists a graph $G'_{\eta, \nu} \in \calPLG^{\simp}[\ell]$
such that $\hom_\eta(G'_{\eta, \nu}, H) \ne \hom_\nu(G'_{\eta, \nu}, H')$.
Let $S$ be the set
consisting of these graphs $G_{\eta, \mu}$ and $G'_{\eta, \nu}$
(we remove any repetitions). 
Note that if $\mu' \in I$, then $\hom_\eta(G, H) = \hom_{\mu'}(G, H)$ for any $G \in \calPLG[\ell]$,
and if $\nu' \in J$, then $\hom_\eta(G, H) = \hom_{\nu'}(G, H')$ for any $G \in \calPLG[\ell]$.
In particular, both equalities hold for each $G \in S$.
We impose a linear order on $S$ and regard any set indexed by $S$ as a tuple.
For any tuple $\bar h = (h_G)_{G \in S}$ of integers,
 where each $0 \le h_G < 2 m^\ell$,
consider the graph $G_{\bar h} = \prod_{G \in S} G^{h_G} \in \calPLG^{\simp}[\ell]$.
Substituting $G = G_{\bar h}$ in~(\ref{eq:gh-expanded-along-labels})
and using the multiplicativity of partial graph homomorphisms (\ref{eq:part-hom-multiplicativity}), we obtain
\[
\sum_{\substack{\mu \colon [\ell] \to [m] \\ \mu \text{ extends } \varphi }}\!\! \Bigg(\prod_{k < i \le \ell} \alpha_{\mu(i)}\Bigg)\! \prod_{G \in S} (\hom_\mu(G, H))^{h_G} \!=\!
\sum_{\substack{\nu \colon [\ell] \to [m'] \\ \nu \text{ extends } \psi}}\!\! \Bigg(\prod_{k < i \le \ell}\alpha'_{\nu(i)}\Bigg)\! \prod_{G \in S} (\hom_\nu(G, H'))^{h_G}
\]
for every $0 \le h_G < 2 m^\ell$.
By the previous observations and the fact that $S$ contains
each $G_{\eta, \mu}$ and $G'_{\eta, \nu}$, the tuple $(\hom_\eta(G, H))_{G \in S}$ coincides with
the tuple $(\hom_\mu(G, H))_{G \in S}$ for $\mu \in I$ and it also coincides with
the tuple $(\hom_\nu(G, H'))_{G \in S}$ for $\nu \in J$;
on the other hand, this  tuple $(\hom_\eta(G, H))_{G \in S}$ 
is different from
the tuple $(\hom_\mu(G, H))_{G \in S}$ for each $\mu \colon [k] \to V(H)$ extending $\varphi$ and not in $I$ and it is also different from
the tuple $(\hom_\nu(G, H'))_{G \in S}$ for each $\nu \colon [k] \to V(H')$ extending $\psi$ and not in $J$.
Transferring the RHS to the LHS and then applying Corollary~\ref{cor:vand-cancellation-gen},
we conclude that
\begin{equation}\label{eq:alpha-part-prod}
\sum_{\mu \in I} \Bigg(\prod_{k < i \le \ell} \alpha_{\mu(i)}\Bigg) = \sum_{\nu \in J} \Bigg(\prod_{k < i \le \ell} \alpha'_{\nu(i)}\Bigg).
\end{equation}
If $\mu \in I$, then $\mu = \sigma \circ \eta$ for some $\sigma \in \Aut(H)$,
and therefore $\alpha_{\mu(i)} = \alpha_{\sigma(\eta(i))} = \alpha_{\eta(i)}$ for $1 \le i \le \ell$.
Hence $\prod_{k < i \le \ell} \alpha_{\mu(i)} = \prod_{k < i \le \ell} \alpha_{\eta(i)}$ (if $k = \ell$, then both sides are $1$),
so (\ref{eq:alpha-part-prod}) transforms to
\begin{equation}\label{eq:alpha-part-prod-simp}
|I|_\bbF \cdot \prod_{k < i \le \ell} \alpha_{\eta(i)} 
= \sum_{\nu \in J} \Bigg(\prod_{k < i \le \ell} \alpha'_{\nu(i)}\Bigg).
\end{equation}
Here we let $|I|_\bbF = |I| \cdot 1_\bbF = 1_\bbF + \ldots + 1_\bbF \in \bbF$ ($1_\bbF$ occurs $|I|$ times).
Since all $\alpha_i \ne 0$, we have $\prod_{k < i \le \ell} \alpha_{\eta(i)} \ne 0$
(if $k = \ell$, this product is $1_\bbF \ne 0$).
Because $I \ne \emptyset$ (as $\eta \in I$) we have $|I| \ge 1$;
 but $\Char \bbF = 0$ so $|I|_\bbF \ne 0$
and therefore $|I|_\bbF \cdot \prod_{k < i \le \ell} \alpha_{\eta(i)} \ne 0$.
This implies that the RHS of (\ref{eq:alpha-part-prod-simp}) is nonzero as well
which can only occur when $J \ne \emptyset$.
Take $\xi \in J$. Then $\xi = \sigma \circ \eta$ for some
isomorphism of $\bbF$-weighted graphs $\sigma \colon V(H) \to V(H')$ from $H$ to $H'$.
Restricting to $[k]$, we obtain
 $\psi  = \sigma \circ \varphi$
which completes the proof.
\end{proof}
\begin{remark}\label{effective-issue-inmain-thm}
Theorem~\ref{thm:gh-main} shows that the
condition  $\hom_\varphi(G, H) = \hom_\psi(G, H')$ 
for all $G \in \calPLG^{\simp}[k]$ is equivalent to the existence of
an isomorphism from $H$ to $H'$
such that $\psi = \sigma \circ \varphi$.  This  condition is effectively
checkable. However, for the purpose of effectively  producing
a \#P-hardness reduction in the dichotomy theorems, e.g., in~\cite{Cai-Chen-Lu-2013},
 we need
a witness $G$ such that $\hom_\varphi(G, H) \not = \hom_\psi(G, H')$.

The proof of
 Theorem~\ref{thm:gh-main} 
gives  an \textit{explicit finite} list to check.
For $\ell=k + 2m^3$,
let $\mathcal P_{\ell,m} = 
\{ \prod_{G \in S} G^{h_G} \mid 0 \le h_G < 2 m^{\ell}\}$,
where $S$ is from the proof of Theorem~\ref{thm:gh-main} using the
construction of Lemma~\ref{lem:gh-highly-surj}.
We then apply $\pi_{[k]}$ to each graph in $\mathcal P_{\ell,m}$
and  define $\mathcal Q_{k, m} = \pi_{[k]}(\mathcal P_{\ell,m}) \subseteq \calPLG^{\simp}[k]$.
Then the proof of Theorem~\ref{thm:gh-main} shows that
the existence of an isomorphism $\sigma \colon V(H) \to V(H')$ from $H$ to $H'$ such that $\varphi = \sigma \circ \psi$
is equivalent to $\hom_\varphi(G, H) = \hom_\psi(G, H')$ for every $G \in \mathcal Q_{k, m}$.
\end{remark}


\section{Proof of Main Theorem}\label{sec:MainProof-dir}

The proof in Section~\ref{sec:MainProof} can be easily adapted
to directed graphs to prove Theorem~\ref{thm:gh-main} in full
generality.
Again we assume $V(H) = [m],\, V(H') = [m']$,
and  $(\alpha_i)_{i \in [m]} \in (\bbF \setminus \{0\})^m$, $(\beta_{i j})_{i, j \in [m]} \in \bbF^{m \times m}$,
$(\alpha'_i)_{i \in [m']} \in (\bbF \setminus \{0\})^{m'}$, and $(\beta'_{i j})_{i, j \in [m']} \in \bbF^{m' \times m'}$
are the vertex and edge weights in $H$ and $H'$, correspondingly.
The technical condition of
``super surjectivity'' in Lemma~\ref{lem:gh-highly-surj-dir} 
will be removed later.

\begin{lemma}\label{lem:gh-highly-surj-dir}
Let $H, H'$ be directed  $\bbF$-weighted graphs
such that $H$ is twin-free and
$m \ge m'$.
Suppose 
 $\varphi \colon [k] \to V(H)$ and $\psi \colon [k] \to V(H')$ where $k \ge 0$.
Assume $\varphi$ is ``super surjective'', namely:
 $|\varphi^{-1}(u)| \ge 4m^2$ for every $u \in V(H)$.
If $\hom_\varphi(G, H) = \hom_\psi(G, H')$ for every $G \in \calPLG^{\simp}[k]$, 
then there exists an isomorphism $\sigma$ from $H$ to $H'$
such that $\psi = \sigma \circ \varphi$.
\end{lemma}
\begin{proof}
Assume $m \ge 1$ (the case  $m = 0$ is trivial). Taking any $u \in V(H)$, we get
$k \ge |\varphi^{-1}(u)| \ge 4 m^2 >0$ and thus $m' \ge |\psi([k])| \ge 1$.
For each $\kappa = (b_i)_{i \in [k]} \in \{\downarrow, \uparrow, \perp\}^k$,
we define a directed graph $G_\kappa \in \calPLG^{\simp}[k]$:
\begin{quote}
$V(G_\kappa) = \{u_1, \ldots, u_k, v\}$,
with each $u_i$ labeled $i$.
For each $i \in [k]$, there is a directed edge 
from $v$ to $u_i$ if $b_i = \downarrow$,
a directed edge from $u_i$ to $v$ if $b_i = \uparrow$ and 
no edge between $u_i$ and $v$ if $b_i = \perp$.
There are no other edges.
\end{quote}

We now define a specific set of $G_\kappa$, given $\varphi$ and $\psi$.
We can partition $[k] = \bigsqcup_{i = 1}^m I_i$
where each $I_i = \varphi^{-1}(i)$ and
$|I_i| \ge 4m^2$.
For every $i \in [m]$, since $|I_i| \ge 4m m'$,
there exists  
$J_i \subseteq I_i$ such that $|J_i| \ge 4 m >0$
and the restriction $\psi_{|J_i}$ takes a constant value $s(i)$,
for some function $s \colon [m] \to [m']$.
Next, for each $i \in [m]$ and
for every  $0 \le k_i, \ell_i < 2m$,
we can fix disjoint subsets $K_i\subset J_i$ and $L_i \subset J_i$,
with $|K_i| = k_i$ and $|L_i| = \ell_i$.
Then
we let the tuple $\chi = \chi(K_1, \ldots, K_m, L_1, \ldots, L_m)
\in \{\downarrow, \uparrow, \perp\}^k$
take $\chi_{|K_i} = \downarrow$, $\chi_{|L_i} = \uparrow$,  $i \in [m]$,  and all other entries are $\perp$.
 This way every possible
 choice of $(k_i, \ell_i)_{i \in [m]}$ with $0 \le k_i, l_i < 2m$
defines a tuple of sets $(K_i, L_i)_{i \in [m]}$ with disjoint $K_i, L_i 
\subset J_i$ for $i \in [m]$,
which in turn defines the tuple $\chi 
 \in \{\downarrow, \uparrow, \perp\}^k$ by the previous rule. 
Let $R$ be the set of all such tuples $\chi$
 for every possible choice of 
$(k_i, \ell_i)_{i \in [m]}$ with $0 \le k_i, l_i < 2m$. 

Then $\hom_\varphi(G_\chi, H) = \hom_\psi(G_\chi, H')$
for every $G_\chi$ with $\chi \in R$ is expressed by: For all
$0 \le k_j, \ell_j < 2m$, $j \in [m]$,
\begin{equation}\label{eq:hom-G-kappa-1st-dir}
\sum_{i = 1}^m \alpha_i \prod_{j = 1}^m (\beta_{i j}^{k_j} \beta_{j i}^{\ell_j}) = \sum_{i = 1}^{m'} \alpha'_i \prod_{j = 1}^m \left( (\beta'_{i s(j)})^{k_j} (\beta'_{s(j) i})^{\ell_j} \right).
\end{equation}
In (\ref{eq:hom-G-kappa-1st-dir}) the sums on $i$ 
come from  assigning  $v \in V(G_\chi)$
 to $i \in V(H)$ or to $i \in V(H')$, respectively.

Because $H$ is twin-free the $2 m$-tuples $(\beta_{i j}, \beta_{j i})_{j \in [m]} \in \bbF^{2 m}$
for $i \in [m]$ are pairwise distinct.
In (\ref{eq:hom-G-kappa-1st-dir}) the sum in the LHS 
has $m$ terms, while the sum in the RHS has  $m' \le m$ terms.
Transferring the RHS to the LHS we get
at most $2m$ terms.
Now we apply Corollary~\ref{cor:vand-cancellation-gen} to the sum
obtained by moving all terms of the RHS to the LHS in (\ref{eq:hom-G-kappa-1st-dir}).
By the pairwise distinctness of
the $2 m$-tuples $(\beta_{i j}, \beta_{j i})_{j \in [m]}
\in \bbF^{2 m}$ for $i \in [m]$, 
and since there are only $m' \le m$ terms from the RHS and every $\alpha_i \ne 0$,
we see that each term from the LHS of
(\ref{eq:hom-G-kappa-1st-dir}) must be canceled by exactly one term from the RHS. 
And this can only occur if $m = m'$,
and there is a bijective map $\sigma \colon [m] \to [m]$:
\begin{equation}\label{eq:alpha-beta-cancellation-itentities-dir}
\alpha_i = \alpha'_{\sigma(i)} \quad \text{for } i \in [m],
\quad  \quad (\beta_{i j}, \beta_{j i})_{j \in [m]} = (\beta'_{\sigma(i) s(j)}, \beta'_{s(j) \sigma(i)})_{j \in [m]} \quad \text{for } i \in [m].
\end{equation} 
Note that since (\ref{eq:alpha-beta-cancellation-itentities-dir})
holds for all $i, j \in [m]$,  if we observe the second entry of each
pair, we have $\beta_{i j} = \beta'_{s(i) \sigma(j)}$
(as well as $\beta'_{\sigma(i) s(j)}$ from the first entry of each pair) for all $i, j \in [m]$.

Next we show that $s$ is bijective. If for some $x, y \in [m]$ we have $s(x) = s(y)$,
then
\[
(\beta_{x j}, \beta_{j x})_{j \in [m]} 
= (\beta'_{s(x) \sigma(j)}, \beta'_{\sigma(j) s(x)})_{j \in [m]} 
= (\beta'_{s(y) \sigma(j)}, \beta'_{\sigma(j) s(y)})_{j \in [m]} 
= (\beta_{y j}, \beta_{j y})_{j \in [m]}.
\]
Since $H$ is twin-free, we conclude that $x = y$. Thus 
the map $s \colon [m] \to [m]$ 
is injective and so (since $[m]$ is finite) $s$ is bijective.
However, $\sigma \colon [m] \to [m]$ is also bijective, 
so it follows from (\ref{eq:alpha-beta-cancellation-itentities-dir})
that the tuples $(\beta'_{i j}, \beta'_{j i})_{j \in [m]}$ for $i \in [m]$
are pairwise distinct, which means that $H'$ is twin-free as well.

Next we show $\psi_{|I_i} = s(i)$ for all $i \in [m]$.
If for all   $i \in [m]$, we have  $J_i = I_i$, then we are done.
Otherwise, take any $w \in [m]$ such that $J_w$ is a proper subset of
$I_w$ and 
we take any $t \in I_w \setminus J_w$.
Observe that $t \notin K_i \cup L_i$ for all $i \in [m]$.
In particular, $\chi(t) = \perp$ for each $\chi \in R$.

For each $\chi \in R$ and  $b \in \{\downarrow, \uparrow\}$,
let $\chi_b$ be the tuple obtained from $\chi$ by reassigning $\chi(t)$
(changing its $t$-th entry) from $\perp$ to $b$
and let $R_{b}$ be the set of all such $\chi_b$.

Then $\hom_\varphi(G_{\kappa}, H) = \hom_\psi(G_{\kappa}, H')$
for every $G_{\kappa}$ with $\kappa \in R_{\downarrow}$ is expressed as
(recall that we have already proved that $m'=m$)
\[
\sum_{i = 1}^m \alpha_i \beta_{i w} \prod_{j = 1}^m (\beta_{i j}^{k_j} \beta_{j i}^{\ell_j}) = \sum_{i = 1}^m \alpha'_i \beta'_{i \psi(t)} \prod_{j = 1}^m \left( (\beta'_{i s(j)})^{k_j} (\beta'_{s(j) i})^{\ell_j} \right),
\]
which can be compared to (\ref{eq:hom-G-kappa-1st-dir}), and here
for
$\kappa \in R_{\downarrow}$
we have one extra edge $(v,u_t)$ in $G_{\kappa}$,
and $\varphi(t) =w$ since $t \in I_w$.
So this holds for every $0 \le k_j, \ell_j < 2m$ where $j \in [m]$.
Transferring the RHS to the LHS and using (\ref{eq:alpha-beta-cancellation-itentities-dir}), we get
\[
\sum_{i = 1}^m (\alpha_i \beta_{i w} - \alpha'_{\sigma(i)} \beta'_{\sigma(i) \psi(t)}) \prod_{j = 1}^m (\beta_{i j}^{k_j} \beta_{j i}^{\ell_j}) = 0
\]
for every $0 \le k_j, \ell_j < 2m$ where $j \in [m]$.
Since $\alpha_i = \alpha'_{\sigma(i)} \not = 0$,
and the tuples $(\beta_{i j}, \beta_{j i})_{j \in [m]}$ for $i \in [m]$ are pairwise distinct,
by Corollary~\ref{cor:vand-cancellation-gen},
we get $\beta_{i w} -  \beta'_{\sigma(i) \psi(t)} = 0$ for $i \in [m]$.
Using $\hom_\varphi(G_{\kappa}, H) = \hom_\psi(G_{\kappa}, H')$
for every $G_{\kappa}$ with $\kappa \in R_{\uparrow}$,
we similarly conclude that $\beta_{w i} = \beta'_{\psi(t) \sigma(i)}$ for $i \in [m]$.
On the other hand, by~(\ref{eq:alpha-beta-cancellation-itentities-dir}),
$\beta_{i w} = \beta'_{\sigma(i) s(w)}$  and $\beta_{w i} = \beta'_{s(w) \sigma(i)}$ for $i \in [m]$.
It follows that $\beta'_{\sigma(i) \psi(t)} = \beta'_{\sigma(i) s(w)}$ and $\beta'_{\psi(t) \sigma(i)} = \beta'_{s(w) \sigma(i)}$
for $i \in [m]$.
However, since $\sigma \colon[m] \to [m]$ is a bijection and,
as shown before, $H'$ is twin-free, this implies that  $\psi(t) = s(w)$. 
Recall that  $\psi_{|J_w} = s(w)$.
This proves that on $I_w \setminus J_w$,  $\psi$ also takes
the constant value $s(w)$.  
Thus $\psi_{|I_i} = s(i)$ for all $i \in [m]$.

Next we prove that $\beta_{i j} = \beta'_{\sigma(i) \sigma(j)}$ for
all $i, j \in [m]$,
i.e., $\sigma$ preserves the edge weights.

For each $\lambda = (b_i)_{i \in [k]}$,
$\tau = (c_i)_{i \in [k]} \in \{\downarrow, \uparrow, \perp\}^k$,
we define a directed graph $G_{\lambda, \tau}  \in \calPLG^{\simp}[k]$:
\begin{quote}
$V(G_{\lambda, \tau}) = \{u_1, \ldots, u_k, v, v'\}$,
with each $u_i$ labeled $i$.
There is a directed edge from $v$ to $v'$ and,
for each $i \in [k]$, there is a directed edge
from $v$ to $u_i$ if $b_i = \downarrow$,
from $u_i$ to $v$ if $b_i = \uparrow$,
and no edge between $u_i$ and $v$ if $b_i = \perp$;
there is a directed edge from $v'$ to $u_i$ if $c_i = \downarrow$,
from $u_i$ to $v'$ if $c_i = \uparrow$,
and no edge between $u_i$ and $v'$ if $c_i = \perp$.
There are no other edges.
\end{quote}

Let $R^2 = R \times R$.
By the definition of $R$, every $(\lambda, \tau) \in R^2$
corresponds to a sequence of tuples $(K_i, L_i, K_i', L_i')$, $i \in [m]$,
where $K_i \sqcup L_i \subset J_i
\subseteq I_i$, and also $K'_i \sqcup L'_i \subset J_i
\subseteq I_i$, for $i \in [m]$
such that $|K_i| = k_i, |L_i| = \ell_i,  |K'_i| = k'_i, |L'_i| = \ell'_i$,
where $0 \le k_i, \ell_i,  k'_i, \ell'_i < 2m$,
 and 
$\lambda_{|K_i}  = \downarrow$,  $\lambda_{|L_i} = \uparrow$, 
$\tau_{|K'_i}  =  \downarrow$,  $\tau_{|L'_i} = \uparrow$,
 and all other entries are $\perp$.
Moreover, the above correspondence between $(\lambda, \tau) \in R^2$
and the sequences of tuples $(K_i, L_i, K_i', L_i')$, $i \in [m]$, is bijective.
In particular, every possible choice of $0 \le k_i, \ell_i, k'_i, \ell'_i < 2m$, $i \in [m]$, is accounted for exactly once under this correspondence.

Then $\hom_\varphi(G_{\lambda, \tau}, H) = \hom_\psi(G_{\lambda, \tau}, H')$
for every $G_{\lambda, \tau}$ with $(\lambda, \tau) \in R^2$  is
expressed as
\[
\sum_{i, j = 1}^m \alpha_{i} \alpha_{j} \beta_{i j} \prod_{r = 1}^m (\beta_{i r}^{k_r} \beta_{r i}^{\ell_r} \beta_{jr}^{k'_r} \beta_{rj}^{\ell'_r}) =
\sum_{i, j = 1}^m \alpha'_{i} \alpha'_{j} \beta'_{i j} \prod_{r = 1}^m \left( (\beta'_{i s(r)})^{k_r} (\beta'_{s(r) i})^{\ell_r} (\beta'_{j s(r)})^{k'_r} (\beta'_{s(r) j})^{\ell'_r} \right).
\]
This holds for all $0 \le k_r, \ell_r, k'_r, \ell'_r < 2m$ 
where $r \in [m]$.
Transferring the RHS to the LHS and using (\ref{eq:alpha-beta-cancellation-itentities-dir}), we get
\[
\sum_{i, j = 1}^m (\alpha_{i} \alpha_{j} \beta_{i j} - \alpha_{i} \alpha_{j} \beta'_{\sigma(i) \sigma(j)}) \prod_{r = 1}^m (\beta_{i r}^{k_r} \beta_{r i}^{\ell_r} \beta_{jr}^{k'_r} \beta_{rj}^{\ell'_r}) = 0,
\]
for every $0 \le k_r, \ell_r,  k'_r, \ell'_r < 2m$ where $r \in [m]$.
Since the tuples $(\beta_{i r}, \beta_{r i})_{r \in [m]}$ for $i \in [m]$ are pairwise distinct,
the tuples $(\beta_{i r}, \beta_{r i}, \beta_{j r}, \beta_{r j})_{r \in [m]}$ for $i, j \in [m]$ are also pairwise distinct.
Then by Corollary~\ref{cor:vand-cancellation-gen}, $\alpha_{i} \alpha_{j} \beta_{i j} - \alpha_{i} \alpha_{j} \beta'_{\sigma(i) \sigma(j)} = 0$ for $i, j \in [m]$.
But all $\alpha_i \ne 0$, so
\begin{equation}\label{eq:beta-preservation-dir}
\beta_{i j} = \beta'_{\sigma(i) \sigma(j)} \quad \text{for } i, j \in [m].
\end{equation}
This 
means that the bijection $\sigma \colon [m] \to [m]$ preserves the edge weights in addition to the vertex weights by~(\ref{eq:alpha-beta-cancellation-itentities-dir}).
Hence $\sigma \colon [m] \to [m]$ is an isomorphism of $\bbF$-weighted graphs from $H$ to $H'$.

Finally, we show that $\psi = \sigma \circ \varphi$.
From~(\ref{eq:alpha-beta-cancellation-itentities-dir}) and~(\ref{eq:beta-preservation-dir}), we have
$\beta'_{\sigma(i) s(j)} = \beta_{i j} = \beta'_{\sigma(i) \sigma(j)}$ 
and $\beta'_{s(j) \sigma(i)} = \beta_{j i} = \beta'_{\sigma(j) \sigma(i)}$ for $i, j \in [m]$.
As $H'$ is twin-free and $\sigma$ is bijective we get $\sigma(j) = s(j)$ for $j \in [m]$.
Now let $x \in [k]$. Then $x \in I_i$ for some $i \in [m]$.
Therefore $\varphi(x) = i$ and so $\psi(x) = s(i) = \sigma(i) = \sigma(\varphi(x))$ confirming $\psi  = \sigma \circ \varphi$.
\end{proof}

With Lemma~\ref{lem:gh-highly-surj-dir}  
replacing Lemma~\ref{lem:gh-highly-surj},
the proof of Theorem~\ref{thm:gh-main} 
for undirected graphs in Section~\ref{sec:MainProof} can be easily adapted 
to directed graphs. This completes the proof of Theorem~\ref{thm:gh-main}.

\begin{remark}
The same remark at the end of Section~\ref{sec:MainProof} holds
for directed graphs with $\ell = k + 4m^3$.
\end{remark}


\begin{proof}[Proof of Corollary~\ref{cor:gh-only-ew-k=0}]
We prove a stronger statement below; Corollary~\ref{cor:gh-only-ew-k=0}
is the special case $k=0$ of Corollary~\ref{cor:gh-only-ew}.
\begin{corollary}\label{cor:gh-only-ew}
Let $H, H'$ be (directed or undirected) $\bbF$-edge-weighted graphs.
Let $\varphi \colon [k] \to V(H)$ and $\psi \colon [k] \to V(H')$ where $k \ge 0$.
If $\hom_\varphi(G, H) = \hom_\psi(G, H')$ for every $G \in \calPLG^{\simp}[k]$, 
then there exists an isomorphism
$\sigma \colon V(H) \to V(H')$ from $H$ to $H'$
such that $\psi' = \sigma \circ \varphi$,
where for every $i \in [k]$, $\psi'(i)$ is a twin of $\psi(i)$.
\end{corollary}
%
%
%
%
Since all vertex weights of $H$ and $H'$ are $1$,
after the twin reduction steps of $H$ to $\widetilde H$ and of $H'$ to $\widetilde{H'}$
(there is no $0$-weight vertices removal),
every vertex $u \in V(H)$ is associated
 to a vertex $\widetilde u \in V(\widetilde H)$
with a positive integer vertex weight equal to the number of the vertices 
in the twin equivalence class in $V(H)$ containing $u$;
the same is true for every vertex in $V(H')$.
We  define $\widetilde \varphi \colon [k] \to V(\widetilde{H})$
 by $\widetilde{\varphi}(i) = \widetilde{\varphi(i)}$,
and $\widetilde \psi \colon [k] \to V(\widetilde{H'})$ from $\psi$
similarly.
It follows that $\hom_{\widetilde \varphi}(G, \widetilde H) = \hom_{\widetilde \psi}(G, \widetilde{H'})$ for every $G \in \calPLG^{\simp}[k]$.
By construction, the graphs $\widetilde H$ and $\widetilde{H'}$ are twin-free.
Therefore by Theorem~\ref{thm:gh-main},
there is an isomorphism
of $\bbF$-weighted graphs $\xi \colon V(\widetilde H) \to V(\widetilde{H'})$ 
from $\widetilde H$ to $\widetilde{H'}$
such that $\widetilde \psi = \xi \circ \widetilde \varphi$.
In particular, $\xi$ preserves the vertex weights, so that the twin class
$\widetilde{u}$ has the same size as
its corresponding $\xi(\widetilde{u})$.
By $\widetilde \psi = \xi \circ \widetilde \varphi$
the subset of labels in $[k]$ mapped  by $\varphi$ to a twin class 
is the same  subset mapped by $\psi$ to  the corresponding twin class. 
We can fix a bijection
for each pair of corresponding twin classes.
Then we can define
an isomorphism $\sigma$ of $\bbF$-weighted graphs from
$H$ to $H'$ by mapping $V(H)$ to $V(H')$, starting with
 $\xi$ and then expanding within each pair of corresponding
twin classes according to the chosen bijection.
The map $\psi' = \sigma \circ \varphi$
can only differ from $\psi$ mapping $i \in [k]$ to
a twin of $\psi(i)$, for all $i \in [k]$.
%
\end{proof}

\begin{remark}\label{remark-on-regts}
In their proof of Theorem 2.1 in~\cite{Goodall-et-al},
Goodall, Regts and Vena noted in a footnote on p.~268 that
``Even though Lov\'{a}sz [9] (this is~\cite{Lovasz-2006})
 works over $\mathbb{R}$ and assumes $a_i > 0$ for all $i$, 
it is easy to check that the proof of
his Lemma 2.4 remains valid in our setting.''
 The following is a sketch of this approach due to Goodall, Regts and Vena;
we thank Guus Regts for enlightening discussions.
Notations below can be found in~\cite{Lovasz-2006}.

Following the proof of Lemma 2.4 by  Lov\'{a}sz in~\cite{Lovasz-2006}
one comes to Claim 4.2. On p.~968 the set $\Psi$ is defined as
consisting of  all maps
equivalent to $\mu$. This set appears in the sum on line 2 of p.~969.
Let $\Psi_0$ be
the subset of   all maps $\eta \in \Psi$
that restrict to $\phi$. Since  $\mu \in \Psi$ and $\mu$ restricts to
$\phi$, the subset $\Psi_0$ is nonempty; however it could have
 cardinality $> 1$. Then in the sum on line~2 of p.~969
when one collects all terms with the same restriction, 
$\phi$ appears with the coefficient
$\sum_{\eta \in \Psi_0} \alpha(\eta(k + 1))$.
The crucial step in the proof of Claim 4.2
is that this sum is nonzero.
 In~\cite{Lovasz-2006} this is nonzero since all vertex weights
 are positive.
Without this positivity, it is possible that such a sum is $0$.
Thus one cannot directly follow the proof in~\cite{Lovasz-2006}
in the general case when a nonempty subset of vertex weights can sum to $0$.

However, instead of Claim 4.2 in~\cite{Lovasz-2006}
on line 6 of p.~970 after the proof of
Claim 4.4, one can
 extend $\phi$ to a  surjective map
$\mu: [\ell] \rightarrow [m]$ for some $\ell \ge k$. 
Then apply the trace operator of~\cite{Lovasz-2006}
 $\ell -k$ times to an analogous sum defined on line $-1$ of p.~968, 
which is in 
${\cal A}''_{\ell}$. 
This gives a sum $\Sigma = \sum_{\eta \in \Psi} \prod_{i \in [k+1: \ell]} 
\alpha(\eta(i)) \eta_{\rm res}$,
where $\eta_{\rm res}$ is the restriction of $\eta$ to $[k]$.

Then use Claim 4.4, all $\alpha(\eta(i))$ are the same for $\eta \in \Psi$
and thus $P = \prod_{i \in [k+1: \ell]} \alpha(\eta(i))$ is a common
factor. Since each $\alpha(\eta(i)) \ne 0$, this $P \ne 0$.
Since $\mu \in \Psi$, $\phi = \mu_{\rm res}$ appears as a term in $\Sigma$
with a coefficient $|S|P$, where $S$ is the subset of $\Psi$ consisting of maps
that restrict to $\phi$.  As $\mathbb{F}$ has characteristic $0$, 
this $|S|P \ne 0$.

Now since $\psi$ is given as equivalent to $\phi$,
and the sum $\Sigma$ is in ${\cal A}''_k$,
$\psi$ must have the same coefficient as that of $\phi$, which has
the form $|S'|P$, where $S'$ is the subset of $\Psi$ that restrict to $\psi$.
So $|S'| = |S|$, and in particular, $S'$ is nonempty.
Thus there is some $\eta \in \Psi$ such that $\eta$ extends $\psi$,
and $\eta$ is equivalent to $\mu$.

This leads to the following theorem following the approach due to Goodall, Regts and Vena.
\begin{theorem}
Let $\bbF$ be a field of characteristic $0$. Let 
$H$ be an undirected $\bbF$-weighted twin-free graph.
For any $k \ge 0$, if $\phi, \psi \colon [k] \rightarrow V(H)$
and $\hom_\phi(G, H) = \hom_\psi(G, H')$
for all undirected graphs $G$, where $G$ may
have multiloops and multiedges,
then there is an automorphism $\sigma$ of $H$ such that
$\phi = \sigma \circ \psi$.
\end{theorem}

\end{remark}

\section{Effective GH dichotomies}

We briefly discuss how to use Theorem~\ref{thm:gh-main} to make
complexity dichotomies for graph homomorphism effective.
A long and fruitful sequence of work~\cite{Dyer-Greenhill-2000,
Dyer-Greenhill-corrig-2004,
Bulatov-Grohe-2005,Grohe-Thurley-2011,Thurley-2009,Goldberg-et-al-2010}
led to the following complexity dichotomy for
weighted graph homomorphisms~\cite{Cai-Chen-Lu-2013} which unifies these
previous ones:
There is a tractability condition ${\cal P}$ such that
for any (algebraic) complex symmetric matrix $H$, if $H$ satisfies ${\cal P}$
then $\hom(\cdot, H)$ is P-time computable, otherwise
there is a P-time reduction
from a canonical \#P-hard problem to $\hom(\cdot, H)$.  
However, in the long sequence of reductions in~\cite{Cai-Chen-Lu-2013}
there are nonconstructive steps; a prominent example is
the first pinning lemma (Lemma 4.1, p.~937). This involves condensing
$H$ by collapsing ``equivalent'' vertices, while introducing vertex weights.
Consider all $1$-labeled graphs $G$.
We say two vertices $u, v \in V(H)$ are ``equivalent''
if $\hom_u(G, H) = \hom_v(G, H)$ where the 
notation $\hom_u(G, H)$ denotes the partial sum of
$\hom(G, H)$ where we restrict to all mappings which map
 the labeled vertex of $G$
to $u$, and similarly for $\hom_v(G, H)$.
This is just the special case $k=1$ in Theorem~\ref{thm:gh-main} (note that we first apply the twin compression step to $H$).
Previously the P-time reduction was proved existentially.
Using Theorem~\ref{thm:gh-main} (see Remark
at the end of Section~\ref{sec:MainProof}),
 this step can be made effective.

There is a finer distinction between making the dichotomy
effective in the sense discussed here versus 
the decidability of the dichotomy.
The dichotomy criterion in~\cite{Cai-Chen-Lu-2013} is decidable 
in P-time (measured in the size of the description of $H$); 
however, according to the proof in~\cite{Cai-Chen-Lu-2013}  which
involves nonconstructive steps,
  when the decision algorithm decides  $\hom(\cdot,  H)$
is \#P-hard it does not produce a reduction.
The results in this paper can constructively produce such a reduction
when $\hom(\cdot,  H)$
is \#P-hard.

Previous versions of Theorem~\ref{thm:gh-main} (e.g., \cite{Lovasz-1967,Lovasz-2006})
show that the above equivalence on $u,v$ for suitably restricted classes
of $H$ can be decided
by testing for graph isomorphism (with pinning).  However, to actually obtain
the promised P-time reduction one has to search for ``witness'' graphs $G$
to $\hom_u(G, H) \not = \hom_v(G, H)$. Having no  graph isomorphism
mapping $u$ to $v$ does not readily yield such a ``witness'' graph $G$,
although an open ended search is guaranteed to find one.
Thus Theorem~\ref{thm:gh-main} gives a double exponential time (in the size of $H$) algorithm
to find a reduction algorithm (specifically to find a reduction
from the problem EVALP to EVAL in~\cite{Cai-Chen-Lu-2013}), 
while directly applying previous versions
of the theorem gives a computable process with no definite time bound.
(But we emphasize that no previous versions of Theorem~\ref{thm:gh-main}
apply to the dichotomy in~\cite{Cai-Chen-Lu-2013}.)

\section{Counterexample for fields of finite characteristic}\label{sec:counterexample}

In Lemmas~\ref{lem:gh-highly-surj} and \ref{lem:gh-highly-surj-dir}, the field $\bbF$ is arbitrary.
By contrast, for Theorem~\ref{thm:gh-main} the proof 
uses the assumption that $\Char \bbF = 0$.
We show that this assumption 
cannot be removed, for any fixed $k$, by an explicit counterexample.
The counterexample
also applies to Corollaries~\ref{cor:gh-both-twin-free-k=0}
to \ref{cor:gh-only-ew-k=0-matrix-version}.

Let $\Char \bbF = p > 0$.
For $n \ge 2$ and $\ell_1 > \ldots > \ell_n >0$,
define an (undirected) $\bbF$-weighted graph 
$H = H_{n, \ell_1, \dots, \ell_n}$ 
with the vertex set $U \cup \bigcup_{i = 1}^n V_i$
where $U = \{u_1, \ldots, u_n\}$ and 
$V_i = \{v_{i, j} \mid  1 \le j \le \ell_i p\}$, for $i \in [n]$,
and the edge set being the union of the edge sets
that form a copy of the complete graph $K_n$ on $U$ and 
$K_{1 + \ell_i p}$ on $\{u_i\} \cup V_i$ for $i \in [n]$.
$H$ is a simple graph with no loops.
To make
$H$ an $\bbF$-weighted graph,
we assign each vertex and edge weight $1$.
(So $H$ is really unweighted.)
It is easy  to see that $H$ is twin-free:
First, any two distinct vertices from $U$ or from the same $V_i$ are not twins
because $H$ is loopless.  
(Note that for vertices $i,j$ 
to be twins in an undirected graph, it is necessary that
if $(i,j)$ is an edge, then the loops $(i,i)$, $(j,j)$ must also be 
present.)
 Second, for any $i \in [n]$,  $u_i \in U$
and any $v \in V_i$ are not twins by $\deg(u_i) > \deg(v)$.
Third, $u_i \in U$ (or any  $v \in V_i$) and any $w \in V_j$, for $j \ne i$,
 are not twins because $w$ has some neighbor in $V_j$
while $u_i$ (or $v$) does not.
Let  $\sigma \in \Aut(H)$ be an automorphism of $H$.
Each vertex $u \in U$ has the property that $u$ has two
neighbors (one in $U$ and one not in $U$) such that
they are not neighbors to each other. 
This property separates $U$ from the rest.
Furthermore $\deg(u_1) > \ldots > \deg(u_n)$.
Therefore $\sigma$  must fix $U$ pointwise.
Then it is easy to see that $\sigma$ must permute within each $V_i$.

For any $\varphi \colon [k] \to U \subset V(H)$ where $k \ge 0$,
we claim that $\hom_\varphi(G, H) = \hom_\varphi(G, K_U)$ for 
every $G \in \calPLG[k]$,
where $K_U$ is the complete graph with the vertex set $U$.

Let ${\mathfrak S}_N$ denote the symmetry group on $N$ letters.
We define a  group action 
of $\prod_{i=1}^n {\mathfrak S}_{\ell_i p}$ on $\{\xi \mid \xi \colon V(G) \to V(H)\}$
which permutes the images of $\xi$ within each of $V_1, \ldots, V_n$,
and fixes $U$ pointwise.
Thus for $g = (g_1, \ldots, g_n) \in \prod_{i=1}^n {\mathfrak S}_{\ell_i p}$,
if $\xi(w) \in V_i$, then $\xi^g(w) = g_i(\xi(w))$.
This  group action partitions all  $\xi$  
into orbits. 
Consider any $\xi  \colon V(G) \to  V(H)$
extending $\varphi$ such that the image  $\xi(V(G)) \not \subseteq
U$. 
Let $\eta$ be in the  orbit of $\xi$. 
The nonzero contributions to $\hom_\xi(G, H)$ and $\hom_\eta(G, H)$
come from either edge weights within $U$, where they are
identical, or within each $\{u_i\} \cup V_i$. Hence by the definition of
the group action, 
$\hom_\xi(G, H) = \hom_\eta(G, H)$.
The stabilizer of $\xi$ consists of those $g$ such that
each $g_i$ fixes the image set $\xi(V(G)) \cap V_i$ pointwise.
Since $\xi(V(G)) \not \subseteq 
U$, the orbit has cardinality,
which is the index of  the stabilizer, divisible by some $\ell_i p$.
In particular it is $0 \bmod p$.
Thus the total contribution from each  orbit is zero in $\bbF$,
except for those  $\xi$ with $\xi(V(G)) \subseteq  U$.
The claim follows.

For $k \ge 1$, we take $H' = H$.
We say that maps  $\varphi, \psi \colon [k] \to U$ have the \emph{same type}
if for every $i, j \in [k]$, $\varphi(i) = \varphi(j)$ iff $\psi(i)
= \psi(j)$. Thus the inverse image sets $\varphi^{-1}(\varphi(i))$
and $\psi^{-1}(\psi(i))$ have the same cardinality for every $i \in [k]$.
It follows that the image sets $\varphi([k])$ and $\psi([k])$,
consisting of
the elements of $U$ having a nonempty inverse image sets
under $\varphi^{-1}$ and $\psi^{-1}$ respectively,
are of the same cardinality. So there is a bijection $\sigma$ of $U$,
mapping $U  \setminus  \varphi([k])$ to $U  \setminus  \psi([k])$,
and  also for every $i \in [k]$,
$(\sigma \circ \varphi) (i) = \psi(i)$.
Thus  $\sigma \in \Aut(K_U)$ and $\sigma \circ \varphi = \psi$.
Take $\varphi, \psi \colon [k] \to U$ such that $\varphi \ne \psi$, and they
have the same type.
For example, we can take
 $\varphi(i) = u_1$ and $\psi(i) = u_2$ for every $i \in [k]$.
Since  $\varphi$ and $\psi$ have the same type,
clearly $\hom_\varphi(G, K_U) = \hom_\psi(G, K_U)$.
For every $G \in \calPLG[k]$,
we have already shown that
$\hom_\varphi(G, K_U) = \hom_\varphi(G, H)$, and since $H'=H$,
$\hom_\psi(G, K_U) = \hom_\psi(G, H')$,
implying that $\hom_\varphi(G, H) = \hom_\psi(G, H')$.
If  Theorem~\ref{thm:gh-main} were to hold for the
field $\bbF$ of $\Char \bbF = p > 0$,
there would be an
automorphism $\widehat{\sigma}
\in \Aut(H)$ such that $\widehat{\sigma} \circ \varphi = \psi$.
But every  automorphism of $H$ must fix $U$ pointwise,
and thus it restricts to the identity map on $U$.
And since $\varphi([k]) \subseteq U$,
we have  $\widehat{\sigma} \circ \varphi = \varphi  \ne \psi$, 
a contradiction.

When $k = 0$, in addition to
$H= H_{n, \ell_1, \ldots, \ell_n}$ we also take $H'
 = H_{n, \ell'_1, \ldots, \ell'_n}$ 
on the vertex set $U \cup \bigcup_{i = 1}^n V'_i$,
where $n \ge 2$, 
$\ell'_1 > \ldots > \ell'_n >0$
and $(\ell'_1, \ldots, \ell'_n) \ne (\ell_1, \ldots, \ell_n)$.
As $k = 0$, the only possible choices are the empty maps
 $\varphi = \emptyset$ and $\psi = \emptyset$,
and $\hom(G, H) = \hom(G, K_U) = \hom(G, H')$
still holds for every $G \in \calPLG[0]$.
However, the same property that
every vertex $u \in U$ has  two
neighbors such that
they are not neighbors to each other separates $U$ from the rest
in both $H$ and $H'$.
Then the monotonicity  $\deg(u_1) > \ldots > \deg(u_n)$
within both $H$ and $H'$ shows that 
any isomorphism from $H$ to $H'$, if it exists,  must fix $U$ pointwise.
Then it is easy to see that $\sigma$ must be a bijection
from $V_i$ in $H$ to the corresponding copy $V_i'$ in $H'$.
This forces 
$(\ell_1, \ldots, \ell_n) =  (\ell'_1, \ldots, \ell'_n)$,
a contradiction.

\section{Rank of connection tensors and dimension of graph algebras}\label{sec:Implications}


The purpose of this section is to extend 
the main results from~\cite{Lovasz-2006}.
These are stated as Theorems~\ref{thm:con-obj-ranks-char-0-field} and~\ref{thm:N(H,k)-matrix-column-space}.

An $\bbF$-valued \textit{graph parameter} is a function from
finite graph isomorphism classes to $\bbF$.
 For convenience, we think of a graph parameter as a function defined on finite graphs and invariant under graph isomorphism. We 
allow multiple edges in our graphs, but no loops, as input to a graph parameter,
following the standard definition~\cite{Freedman-Lovasz-Schrijver-2007,Lovasz-2006}.
A graph parameter $f$ is called \textit{multiplicative}, if for any disjoint union $G_1 \sqcup G_2$ of graphs $G_1$ and $G_2$ we have $f(G_1 \sqcup G_2) = f(G_1) f(G_2)$.
A graph parameter on a labeled graph ignores its labels.
Every weighted graph homomorphism $f_H = \hom(\cdot, H)$ is a multiplicative graph parameter.

A ($k$-labeled, $\bbF$-)\textit{quantum graph} is a finite formal $\bbF$-linear
combination of finite $k$-labeled graphs.
$\calG[k] = \bbF \calPLG[k]$ is the monoid algebra of $k$-labeled $\bbF$-quantum graphs.
We denote by $\calG^{\simp}[k]$ the monoid algebra of \textit{simple} $k$-labeled $\bbF$-quantum graphs;
it is a subalgebra of $\calG[k]$.
$U_k$ is the multiplicative identity
and the empty sum is the additive identity in both $\calG[k]$ and $\calG^{\simp}[k]$.

Let $f$ be any graph parameter. For all integers $k, n \ge 0$, 
we define the following $n$-dimensional array $T(f, k, n) \in \bbF^{(\calPLG[k])^n}$,
which can be identified with $(V^{\otimes n})^*$, the dual space of $V^{\otimes n}$,
where $V = \bigoplus_{\calPLG[k]} \bbF$ is the infinite dimensional vector space with
coordinates indexed by ${\calPLG[k]}$.
The entry of $T(f, k, n)$ at coordinate $(G_1, \ldots, G_n)$ is $f(G_1 \cdots G_n)$; when $n = 0$, 
we define $T(f, k, n)$ to be the scalar $f(U_k)$. The arrays 
$T(f, k, n)$ are symmetric with respect to its coordinates,
i.e., $T(f, k, n) \in \Sym(\bbF^{(\calPLG[k])^n})$.
We denote by $\symrank T(f_H, k, n)$ its symmetric rank (see the Appendix in Section~\ref{sec:appendix}).
Fix $f, k$ and $n$, we call the $n$-dimensional
 array $T(f, k, n)$ the ($k$-th, $n$-dimensional)
 \textit{connection tensor} of the graph parameter $f$. When $n = 2$, 
a connection tensor 
is exactly
a \textit{connection matrix} $M(f, k)$ of the graph parameter $f$ 
studied in \cite{Freedman-Lovasz-Schrijver-2007}, i.e., $T(f, k, 2) = M(f, k)$.

The  empty graph is $K_0 = U_0$.
The proof of the following theorem comes from~\cite{Cai-Govorov-2019};
for convenience of the reader we include a simple proof here, which
also gives the decomposition~(\ref{eqn:T(f_H,k,n)-decomposition}) below.
\begin{theorem}\label{thm:uniform-exponential-rank-boundedness}
For any graph parameter defined by the graph homomorphism $f_H = \hom(\cdot, H)$,
 we have $f_H(K_0) = 1$ and $\symrank T(f_H, k, n) \leqslant |V(H)|^k$ for all $k, n \ge 0$. 
\end{theorem}
\begin{proof}
The first claim is obvious, as an empty product is $1$,
and the sum in~(\ref{eqn:hom-def})
 is over the unique empty map $\emptyset$ which is the only possible map
from  the empty set $V(K_0)$.
 For the second claim
  notice that by~(\ref{eq:part-hom-multiplicativity}) for any $k$-labeled graphs $G_1, \ldots, G_n$ and $\varphi \colon [k] \to V(H)$,
\begin{equation}\label{eqn:partial-hom-mult}
\hom_\varphi(G_1 \cdots G_n, H) = \hom_\varphi(G_1, H) \cdots \hom_\varphi(G_n, H).
\end{equation}
When $n = 0$, this equality is $\hom_\varphi(U_k, H) = 1$
according to~(\ref{eqn:hom-def}), as an empty product is $1$. 

%

By~(\ref{eqn:full-hom-dec}) and~(\ref{eqn:partial-hom-mult}),
for the connection tensor $T(f_H, k, n)$ we have the following decomposition:
\begin{equation}\label{eqn:T(f_H,k,n)-decomposition}
T(f_H, k, n) = \sum_{\varphi \colon [k] \to V(H)} \alpha_\varphi (\hom_\varphi(\cdot, H))^{\otimes n},
\end{equation}
where each $\hom_\varphi(\cdot, H) \in \bbF^{\calPLG[k]}$ and $k, n \ge 0$. Now each $T(f_H, k, n)$ is a linear combination of $|V(H)|^k$ tensor $n$-powers
and therefore $\symrank  T(f_H, k, n) \le q^k$ for $k, n \ge 0$.
\end{proof}

Similarly to~\cite{Freedman-Lovasz-Schrijver-2007}
where an explicit expressibility criterion
involving connection matrices was shown,
the converse of this theorem is also true
and was shown in~\cite{Cai-Govorov-2019}.
The framework in~\cite{Cai-Govorov-2019} is undirected GH.
However, with slight adjustments,
a similar result can be obtained for directed GH
(for more discussions, see Section 6 of the full version of~\cite{Cai-Govorov-2019}).

For graph parameters of the form $f_H = \hom(\cdot, H)$,
where $H$ has positive vertex weights and real edge weights,
the main results of~\cite{Lovasz-2006}
are to compute the rank of the corresponding connection matrices,
and the dimension of graph algebras, etc.
We will prove these results for  arbitrary $\bbF$-weighted graphs 
(without  vertex or edge weight restrictions). Moreover we will prove these for
connection tensors (see~\cite{Cai-Govorov-2019}). 
Below we let $H$ be a (directed or undirected) $\bbF$-weighted graph.

For $k \ge 0$, let $N(k, H)$ be the matrix
whose rows are indexed by maps $\varphi \colon [k] \to V(H)$ 
and columns  are indexed by $\calPLG[k]$,
and the row indexed by $\varphi$ is $\hom_\varphi(\cdot, H)$.
We have a group action of $\Aut(H)$ on the $k$-tuples from $V(H)^k = \{\varphi \colon [k] \to V(H) \}$
by $\varphi \mapsto \sigma \circ \varphi$ for $\sigma \in \Aut(H)$ and $\varphi \colon [k] \to V(H)$.
We use $\orb_k(H)$ to denote the number of its orbits.

As mentioned before, $f_H = \hom(\cdot, H)$ ignores labels on a labeled graph,
so we can think of $f_H$ as defined on $\calPLG[k]$
and then by linearity as defined on $\calG[k]$.
Then we can define the following bilinear symmetric form on $\calG[k]$:
\[
\langle x, y \rangle = f_H(x y),\quad x, y \in \calG[k]. 
\]

Let
\[
\calK_{[k]} = \{ x \in \calG[k] \colon 
\forall y \in \calG[k], \, \langle x, y \rangle = 0 \}.
\]
Clearly, $\calK_{[k]}$ is an ideal in $\calG[k]$, so we can form a quotient algebra $\calG'[k] = \calG[k] / \calK_{[k]}$. 

In order to be consistent with the notation in~\cite{Lovasz-2006},
when $f = f_H = \hom(\cdot, H)$
for an $\bbF$-weighted (directed or undirected) graph $H$,
we let $T(k, n, H) = T(f_H, k, n)$ where $k, n \ge 0$
and $M(k, H) = M(f_H, k)$ where $k \ge 0$.
It is easy to see that $h \in \calK_{[k]}$
iff  $M(k, H) h = 0$.

The following theorem extends
Theorem 2.2, Corollary 2.3 and the results of Section 3 in~\cite{Lovasz-2006}.

\begin{theorem}\label{thm:con-obj-ranks-char-0-field}
Let $\bbF$ be a field.
Let $H$ be a (directed or undirected) $\bbF$-weighted 
graph.
Then $\calG'[k] \cong \bbF^{r_k}$ as isomorphic algebras, where
\[
r_k = \dim \calG'[k] = \symrank T(k, n, H) = \rank T(k, n, H) = \rank M(k, H) 
\le  \rank N(k, H) \le \orb_k(H)
\]
for $k \ge 0$ and $n \ge 2$.
Here $\symrank$ denotes symmetric tensor rank,
and $\rank$ on $T$ and on $M, N$ denote tensor and matrix rank respectively
(see the Appendix in Section~\ref{sec:appendix}).

Furthermore, if  $\bbF$ has characteristic $0$ and  $H$ is twin-free,
then the above quantities are all equal. 
In particular, if $H$ 
has no nontrivial $\bbF$-weighted automorphisms,
then 
 $r_k = |V(H)|^k$.
\end{theorem}

The following theorem  generalizes Lemma 2.5 in~\cite{Lovasz-2006}.
Let $k \ge 0$ be an integer and $H$ be an $\bbF$-weighted graph.
We say that a vector $f \colon V(H)^k \to \bbF$ is
\textit{invariant under the automorphisms} of $H$ 
if $f(\sigma \circ \varphi) = f(\varphi)$ for every $\sigma \in \Aut(H)$ and $\varphi \in V(H)^k$.

\begin{theorem}\label{thm:N(H,k)-matrix-column-space}
Let $\bbF$ be a field of characteristic $0$.
Let $H$ be a (directed or undirected) $\bbF$-weighted twin-free graph.
Then for $k \ge 0$, the column space of $N(k, H)$ consists of precisely
those vectors $f \colon V(H)^k \to \bbF$ that are invariant under the automorphisms of $H$.
Moreover, every such vector can be obtained as
a finite linear combination of the columns of $N(k, H)$
indexed by $\calPLG^{\simp}[k]$.
\end{theorem}
From this theorem, we can immediately conclude the existence of
simple contractors and connectors 
(from $\calPLG^{\simp}[2]$)
 when $\Char \bbF = 0$ for GH functions
(see~\cite{Lovasz-Szegedy-2009} for the definitions).
We now proceed to the proof.
Below, we let $k \ge 0$ be arbitrary.
For $\varphi \in V(H)^k$, we let $\bfx_\varphi = \hom_\varphi(\cdot, H)$; we will think of it as a row-tuple
and use it as a shorthand.
From the proof of Theorem~\ref{thm:uniform-exponential-rank-boundedness}
we have the following decomposition for the connection tensor $T(H, k, n)$ where $n \ge 0$:
\begin{equation}\label{eq:T(k,n,H)-decomposition}
T(k, n, H) = \sum_{\varphi \colon [k] \to V(H)} \alpha_\varphi (\hom_\varphi(\cdot,H))^{\otimes n}. 
\end{equation}
Here $\alpha_\varphi = \prod_{i \in [k]} \alpha_{\varphi(i)}$
for $\varphi \colon [k] \to V(H)$ as defined right after~(\ref{eqn:partial-hom-dec}).
For the connection matrix $M(k, H)$ (viewed as a $2$-dimensional array (tensor)) this means that
\[
M(k, H) = \sum_{\varphi \colon [k] \to V(H)} \alpha_\varphi (\hom_\varphi(\cdot,H))^{\otimes 2}. 
\]


We call $\varphi, \psi \colon [k] \to V(H)$ equivalent, if $\bfx_\varphi = \bfx_\psi$,
i.e., if $\hom_\varphi(G, H) = \hom_\psi(G, H)$ for all $G \in \calPLG[k]$.
Clearly, this is an equivalence relation so it partitions $V(H)^k$
into (nonempty) equivalence classes $I_1, \ldots, I_s$.
Note that $s \ge 1$ as $V(H)^k \ne \emptyset$.
For every $i \in [s]$, fix $\varphi_i \in I_i$.

We show that $\bfx_{\varphi_1}, \ldots, \bfx_{\varphi_s}$ are linearly independent.
Suppose for some $a_1, \ldots, a_s \in \bbF$,
we have $\sum_{i = 1}^s a_i \bfx_{\varphi_i} = 0$.
Then $\sum_{i = 1}^s a_i \hom_{\varphi_i}(G, H) = 0$ for all $G \in \calPLG[k]$.
For every $1 \le i < j \le s$, we fix $G_{i j} \in \calPLG[k]$
such that $\hom_{\varphi_i}(G_{i j}, H) \ne \hom_{\varphi_j}(G_{i j}, H)$.
For any tuple $\bar h = (h_{i j})_{1 \le i < j \le s}$,
where each $0 \le h_{i j} < s$,
consider the graph $G_{\bar h} = \prod_{1 \le i < j \le s} G^{h_{i j}} \in \calPLG^{\simp}[k]$.
(If $s = 1$, then $\bar h$
is the empty tuple so $G_{\bar h} = U_k$).
Substituting $G = G_{\bar h}$ in the previous equality 
and using the multiplicativity of partial graph homomorphisms~(\ref{eq:part-hom-multiplicativity}),
we obtain
\[
\sum_{i = 1}^s a_i \prod_{1 \le j < \ell \le n} (\hom_{\varphi_i}(G_{j \ell}, H))^{h_{j \ell}} = 0.
\]
for every $0 \le h_{j \ell} < s$ where $1 \le j < \ell \le s$.
Clearly,
the tuples $(\hom_{\varphi_i}(G_{j \ell}, H))_{1 \le j < \ell \le s}$
for $i \in [s]$ are pairwise distinct.
Applying Corollary~\ref{cor:vand-cancellation-gen},
we conclude that each $a_i = 0$,
which shows the linear independence of $\bfx_{\varphi_1}, \ldots, \bfx_{\varphi_s}$.

By definition, every row of $N(k, H)$ is precisely
one of $\bfx_{\varphi_1}, \ldots, \bfx_{\varphi_s}$
and each $\bfx_{\varphi_i}$ appears as a row in $N(k, H)$,
so we conclude that $\rank N(k, H) = s$.
Note that if $\varphi, \psi \in V(H)^k$
are in the same orbit of the action of $\Aut(H)$ on $V(H)^k$,
i.e., $\varphi = \sigma \circ \psi$
for some $\sigma \in \Aut(H)$ 
then $\hom_\varphi(\cdot, H) = \hom_\psi(\cdot, H)$
and therefore $\varphi, \psi \in I_i$ for some $i \in [s]$.
It follows that $\rank N(k, H) = s \le \orb_k(H)$.
Then~(\ref{eq:T(k,n,H)-decomposition}) can be written as
\[
T(k, n, H) = \sum_{i \in [s]} \sum_{\varphi \in I_i} \alpha_\varphi (\hom_\varphi(\cdot, H))^{\otimes n} = \sum_{i \in [s]} (\sum_{\varphi \in I_i} \alpha_\varphi) (\hom_{\varphi_i}(\cdot, H))^{\otimes n} = \sum_{i \in [s]} b_i (\hom_{\varphi_i}(\cdot, H))^{\otimes n},
\]
where $n \ge 0$.
Here $b_i = \sum_{\varphi \in I_i} \alpha_\varphi$ for $i \in [s]$.
Now let $J = \{ i \in [s] \mid b_i \ne 0\}$ and put $r = |J| \le s$. 
Then
\begin{equation}\label{eq:T(k,n,H)-compressed}
T(k, n, H) = \sum_{i \in J} b_i (\hom_{\varphi_i}(\cdot, H))^{\otimes n},
\end{equation}
where $n \ge 0$, and
\begin{equation}\label{eq:M(k,H)-compressed}
M(k, H) = \sum_{i \in J} b_i (\hom_{\varphi_i}(\cdot, H))^{\otimes 2}.
\end{equation}
By Corollary~\ref{cor:tensor-powers-linear-independence-nonsymmetric-and-symmetric-rank-and-uniqueness-long},
we have $\rank T(k, n, H) = \symrank T(k, n, H) = r$ for $n \ge 2$
(and also by Lemma~\ref{lem:tensor-powers-linear-independence-nonsymmetric-rank-and-uniqueness-long} for any $n \ge 3$, any expression $T(k, n, H)$
as $\sum_{i = 1}^r a_i \bfy_{i, 1} \otimes \cdots \otimes \bfy_{i, n}$ is a permutation of the sum in (\ref{eq:T(k,n,H)-compressed})).
It is also clear that the matrix rank $\rank M(k, H) = r$.~\footnote{
It can be shown that the matrix rank of a (possibly infinite) matrix $A$
coincides with the tensor rank of $A$. 
For a symmetric matrix $A$, as long as $\Char \bbF \ne 2$,
one can show that $\rank A = \symrank A$.
}
Now we show that $\calG'[k] \cong \bbF^r$.
For this, define the linear map
\[
\Phi \colon \calG[k] \to \bbF^r, \quad \Phi(h) = (\bfx_{\varphi_i}(h))_{i \in J}, \quad h \in \calG[k]. 
\]
By~(\ref{eq:part-hom-multiplicativity}), we have 
\begin{gather*}
\Phi(g h) = \Phi(g) \Phi(h), \quad g, h \in \calG[k], \\
\Phi(U_k) = ( \bfx_{\varphi_i}(U_k))_{i \in J} = \underbrace{(1, \ldots, 1)}_{r = |J| \text{ times}} \in \bbF^r,
\end{gather*}
so $\Phi \colon \calG[k] \to \bbF^r$ is an algebra homomorphism.
We now prove its surjectivity.
Clearly, $\im \Phi$ is a subalgebra of $\bbF^r$.
If $\im \Phi \neq \bbF^r$, then $\dim \im \Phi < r$,
and so as a vector subspace of $\bbF^r$,
$\im \Phi$ has a nonzero annihilator, 
i.e., there exists a nonzero tuple $(a_i)_{i \in J} \in \bbF^r$ such that
$\sum_{i \in J} a_i \bfx_{\varphi_i} (h) = 0$ for any $h \in \calG[k]$.
This implies that $(\sum_{i \in J} a_i \bfx_{\varphi_i}) (h) = 0$ for any $h \in \calG[k]$
and therefore $\sum_{i \in J} a_i \bfx_{\varphi_i} = 0$ 
contradicting the linear independence of 
$\bfx_{\varphi_1}, \ldots \bfx_{\varphi_s}$ (note that $J \subseteq [s]$).
Therefore $\im \Phi = \bbF^r$.
We have shown that $\Phi \colon \calG[k] \to \bbF^r$ is surjective.

Next, we show $\ker \Phi = \calK_{[k]}$.
For this, let $h \in \calG[k]$.
As noted before, $h \in \calK_{[k]}$ iff $M(k, H) h = 0$
which by (\ref{eq:M(k,H)-compressed}) is equivalent to $\sum_{i \in J} b_i \bfx_{\varphi_i}(h) \bfx_{\varphi_i} = 0$.
This,  by the linear independence of $\bfx_{\varphi_i},\, i \in J$,
is equivalent to $\bfx_{\varphi_i}(h) = 0$ for $i \in J$
which is in turn equivalent to $h \in \ker \Phi$.
Hence $\ker \Phi = \calK_{[k]}$.

Then $\Phi \colon \calG[k] \to \bbF^r$ factors through
$\calG[k] / \ker \Phi = \calG[k] / \calK_{[k]} = \calG'[k]$,
inducing an algebra isomorphism
\[
\Phi' \colon \calG[k] \to \bbF^r, \quad \Phi'(h + \calK_{[k]}) = (\bfx_{\varphi_i}(h))_{i \in J}, \quad h \in \calG[k].
\]
Thus $\calG'[k] \cong \bbF^r$.
This proves the first part of 
Theorem~\ref{thm:con-obj-ranks-char-0-field}.

To continue with the proof of Theorem~\ref{thm:con-obj-ranks-char-0-field},
 we assume $\Char \bbF = 0$ and $H$ is twin-free.
By 
Theorem~\ref{thm:gh-main},
if $\varphi, \psi \in I_i$ for some $i \in [s]$,
then there exists $\sigma \in \Aut(H)$ such that $\psi = \sigma \circ \varphi$. 
It follows that the orbits of the action of $\Aut(H)$ on $V(H)^k$
are precisely $I_i$ where $i \in [s]$.
Hence $\rank N(k, H) = s = \orb_k(H)$.
Additionally, in this case we have $\alpha_{\psi(i)} = \alpha_{\sigma(\varphi(i))} = \alpha_{\varphi(i)}$ for $i \in [k]$,
so $\prod_{i \in [k]} \alpha_{\psi(i)} = \prod_{i \in [k]} \alpha_{\varphi(i)}$,
in other words, $\alpha_\psi = \alpha_\varphi$.
Because $\varphi_i \in I_i$ we obtain that $b_i = \sum_{\varphi \in I_i} \alpha_\varphi = |I_i|_\bbF \cdot \alpha_{\varphi_i}$ for $i \in [s]$
(here  $|I_i|_\bbF = |I_i| \cdot 1_\bbF = 1_\bbF + \ldots + 1_\bbF \in \bbF$ where $1_\bbF$ occurs $|I_i|$ times).
Since all $\alpha_i \ne 0$, we have $\alpha_\varphi \ne 0$
(if $k = 0$, this product is $1_\bbF \ne 0$).
Because $I_i \ne \emptyset$ we have $|I_i| \ge 1$;
but $\Char \bbF = 0$ so $|I_i|_\bbF \ne 0$
and therefore $b_i = |I_i|_\bbF \cdot \alpha_{\varphi_i} \ne 0$ for $i \in [s]$.
Thus $J = [s]$ so that $r = s$.
This finishes the proof of Theorem~\ref{thm:con-obj-ranks-char-0-field}.

Next, since every row  of $N(k, H)$ indexed by $\varphi \in I_i$
is the same within $I_i$,
\[
\im N(k, H) 
\subseteq \Span\{ \sum_{\varphi \in I_i} e_\varphi \}_{i \in [s]},
\]
where  $\im N(k, H)$ consists of finite linear combinations of
the columns of $N(k, H)$, 
$\{e_\varphi\}_{\varphi \in (V(H))^k}$ is the canonical basis of $\bbF^{V(H)^k}$,
i.e., the $\varphi$-th entry of $e_\varphi$ is $1$ while the rest are $0$.
We denote the vector space on the RHS of this equation by $V$.
Since $\rank N(k, H) = s$ and $\dim V = s$,
it follows that $\im N(k, H) = V$,
i.e., $V$ is the column space of $N(k, H)$.
As shown before, $I_i$ where $i \in [s]$
are precisely the orbits of the action of $\Aut(H)$ on $V(H)^k$.
Therefore $V$ is precisely the subspace of the vectors in $\bbF^{V(H)^k}$
invariant under the automorphisms of $H$.

Since $\Char \bbF = 0$ and $H$ is twin-free,
in the previous proof of the linear independence of $\bfx_{\varphi_1}, \ldots, \bfx_{\varphi_s}$
by 
Theorem~\ref{thm:gh-main},
we can in fact choose $G_{i j} \in \calPLG^{\simp}[k]$ 
such that $\hom_{\varphi_i}(G_{i j}, H) \ne \hom_{\varphi_j}(G_{i j}, H)$,
for all $1 \le i < j \le s$.
Then this proof further implies
the linear independence of the restrictions of $\bfx_{\varphi_1}, \ldots, \bfx_{\varphi_s}$
to the coordinates $G_{\bar h} = \prod_{1 \le i < j \le n} G_{i j}^{h_{i j}} \in \calPLG^{\simp}[k]$
where $\bar h = (h_{i j})_{1 \le i < j \le s}$ with each $0 \le h_{i j} < s$.
Since every row of $N(k, H)$ is precisely
one of $\bfx_{\varphi_1}, \ldots, \bfx_{\varphi_s}$
and each $\bfx_{\varphi_i}$ appears as a row in $N(k, H)$,
the columns of $N(k, H)$ indexed by these $G_{\bar h}$
form a submatrix of rank $s$,
so we can choose $s$ linearly independent
columns among these.
Since $\rank N(k, H) = s$,
we conclude that these $s$ columns indexed by elements of $\calPLG^{\simp}[k]$ span the entire
column space of $N(k, H)$ which is precisely $V$.
Thus the proof of Theorem~\ref{thm:N(H,k)-matrix-column-space} is complete.

\begin{remark}
When $k = 0$, it is easy to see that
Theorem~\ref{thm:con-obj-ranks-char-0-field} and
Theorem~\ref{thm:N(H,k)-matrix-column-space} hold
for arbitrary fields $\bbF$ and arbitrary $\bbF$-weighted graphs $H$ (directed or undirected).
Indeed,  for $k = 0$,
the quantities in the displayed chain in
Theorem~\ref{thm:con-obj-ranks-char-0-field} are all $1$, and the statements of Theorem~\ref{thm:N(H,k)-matrix-column-space} are vacuously true.
However, when $k \ge 1$, one can show that both assumptions $\Char \bbF = 0$ and $H$ is twin-free are indispensable; otherwise, 
one can construct a counterexample in which all inequalities in
the statement of Theorem~\ref{thm:con-obj-ranks-char-0-field} 
are strict and the statements of Theorem~\ref{thm:N(H,k)-matrix-column-space} are false.  
\end{remark}

\section{Appendix: multilinear algebra}\label{sec:appendix}

In this paper we use notions such as  tensor and symmetric tensor rank
over infinite dimensional vector spaces such as
 $\calG[k]$. This infinite dimensionality
causes some  technical complications, and some well-known
theorems for finite dimensional spaces do not hold.
For the convenience of readers we  recap some definitions and
needed results, see also~\cite{Cai-Govorov-2019}.


Below we use $V, V_1, \ldots $ to denote vector spaces (over $\bbF$)
and $\calI, \calI_1, \ldots$ to denote (index) sets.
If $V$ is a vector space, $V^*$ denotes its dual space,
and  $V^{\otimes n}$ the $n$-fold tensor power of $V$.

The symmetric group ${\mathfrak S}_n$ acts 
on $V^{\otimes n}$  by
$\sigma(\otimes_{i = 1}^n v_i) = \otimes_{i = 1}^n v_{\sigma(i)}$,
and extended by linearity.
Recall that $V^{\otimes n}$ consists of finite linear combinations of
such terms. 
We call a tensor $A \in V^{\otimes n}$ symmetric 
if $\sigma(A) = A$ for all $\sigma \in {\mathfrak S}_n$,
and denote by $\Sym^n(V)$ the set of symmetric tensors in $V^{\otimes n}$.
As $\bbF$ may have finite characteristic $p$,
the usual symmetrizing operator from $V^{\otimes n}$ to $\Sym^n(V)$,
which requires division by $n!$, 
is in general not defined.

We will also denote the space of symmetric 
$n$-fold multilinear functions on $V$ by $\Sym((V^{\otimes n})^*)$, 
i.e., the functions from $(V^{\otimes n})^*$ that are symmetric.
We have $(V^*)^{\otimes n} \cap \Sym((V^{\otimes n})^*) = \Sym^n (V^*)$.



Any $A \in \bigotimes_{j = 1}^n V_j$, where $n \ge 1$, can be expressed as
a finite sum
\[
A = \sum_{i = 1}^r \bfv_{i 1} \otimes \cdots \otimes \bfv_{i n},
\quad \bfv_{i j}\in V_j.
\]
The least $r \ge 0$ for which $A$ has such an expression
is called the
rank of $A$,  denoted by $\rank(A)$.
For $A \in \Sym^n(V)$  the symmetric rank of $A$
is the  least $r \ge 0$ for which $A$ can be expressed as
\[
A = \sum_{i = 1}^r \lambda_i \bfv_i^{\otimes n},
\quad \lambda_i \in \bbF, \bfv_i \in V,
\]
and is denoted by $\symrank(A)$.

We need to refer to the rank of functions
in $(\bigotimes_{i = 1}^n V_i)^*$. Clearly, 
$\bigotimes_{i = 1}^n V_i^*$ canonically embeds as a subspace 
of $(\bigotimes_{i = 1}^n V_i)^*$ (but for infinite dimensional
$V_i$, it embeds as a proper subspace).
For a function $F \in (\bigotimes_{i = 1}^n V_i)^*$, where $n \ge 0$,
we define the rank of the function $F$ to be $\infty$ if $F \notin \bigotimes_{i = 1}^n V_i^*$,
and if $F \in \bigotimes_{i = 1}^n V_i^*$,
the rank of $F$ is the least $r$ for which $F$ can be written as
\[
F = \sum_{i = 1}^r a_i \bff_{i 1} \otimes \cdots \otimes \bff_{i n},
\quad \quad a_i \in \bbF, \quad \bff_{i j} \in V_j^*.
\]
The symmetric rank $\symrank(F)$ of $F \in \Sym((V^{\otimes n})^*)$
  is similarly defined. 
It is $\infty$ if $F \not \in \Sym^n(V^*)$.
For $F \in \Sym^n(V^*)$, we define  $\symrank(F)$  to be the least $r$ 
such that
\[
F = \sum_{i = 1}^r \lambda_i \bff_i^{\otimes n},
\quad \quad \lambda_i \in \bbF, \quad \bff_i \in V^*,
\]
if such an expression exists; $\symrank(F) = \infty$ otherwise.


\begin{lemma}\label{lem:linear-independence-and-nonzero-minor-long}
The vectors $\bfx_1, \ldots, \bfx_r \in \bbF^\calI$ are linearly independent iff in the $r \times \calI$ 
matrix formed by $\bfx_1, \ldots, \bfx_r$ as rows 
there exists a nonzero $r \times r$ minor.
\end{lemma}
\begin{proof}
$\Leftarrow$ is obvious, so let us prove $\Rightarrow$.
Let $R \subseteq [r]$ be a maximal subset satisfying the property that
for some finite subset $C \subseteq \calI$ the
set of vectors $\{\bfx_i\mid_{C} \: : \, i \in R\}$ is linearly independent,
where $\bfx_i\mid_{C}$ is the restriction of $\bfx_i$ to $C$.
Suppose linear independence is achieved by $C$ for $R$.
Then
it also holds for any $C' \supseteq C$. 

If  $R \not = [r]$,  let $j \in [r] \setminus R$,
and consider $R^+ = R \cup \{j\}$. 
$\{\bfx_i\mid_{C}  \: : \, i \in R^+\}$ is linearly dependent.
Hence a unique linear combination holds for some $c_i \in \bbF$ ($i \in R$),
\begin{equation}\label{eqn:onC-lin-dep}
\bfx_j\mid_{C} \, = \sum_{i\in R} c_i \bfx_i\mid_{C}.
\end{equation}
For any $k \not \in C$,
$\{\bfx_i\mid_{C\cup\{k\}}  \: : \, i \in R^+\}$ is
also linearly dependent, and 
we have $\bfx_j\mid_{C\cup\{k\}}
= \sum_{i\in R} c'_i \bfx_i\mid_{C\cup\{k\}}$ for some $c'_i \in \bbF$.
Compared to (\ref{eqn:onC-lin-dep}), $c'_i = c_i$ for all $i \in R$.
Hence $\bfx_j \, = \sum_{i\in R} c_i \bfx_i$, a contradiction
to $\{\bfx_1, \ldots, \bfx_r\}$ being linearly independent.
So $R = [r]$.  There exists a nonzero $r \times r$ minor
in the $R \times C$ submatrix.
\end{proof}

For $\bfx = (x_i)_{i \in \calI} \in \bbF^\calI$
and  $h = (h_i)_{i \in \calI}  \in \bigoplus_\calI \bbF$ 
(in a direct sum, only finitely many $h_i$ are zero),
we denote their dot product
 by $\bfx(h) = \sum_{i \in \calI} x_i h_i \in \bbF$. 
 Here we view $\bbF^\calI$ as the dual space of $\bigoplus_\calI \bbF$.
(In general the dot product for $\bfx, \bfy  \in \bbF^\calI$
is not defined.)

\begin{lemma}\label{lem:linear-independence-and-dual-vectors-long}
Let $\bfx_1, \ldots, \bfx_r \in \bbF^\calI$ be linearly independent. Then there exist $h_1, \ldots, h_r \in \bigoplus_{\calI} \bbF$ dual
to $\bfx_1, \ldots, \bfx_r$, i.e., $\bfx_i(h_j) = \delta_{i j},\, 1 \le i, j 
\le  r$.
\end{lemma}
\begin{proof}
By Lemma~\ref{lem:linear-independence-and-nonzero-minor-long}, there exist $r$ 
distinct indices $k_j \in \calI,\, 1 \le j \le r$ such that the matrix $A = (a_{i j})_{i, j = 1}^r = ((\bfx_i)_{k_j})_{i, j = 1}^r$ is invertible,
and let $B = (b_{i j}) = A^{-1}$.
Taking $h_i = \sum_{j = 1}^r b_{j i} e_{k_j} \in \bigoplus_\calI \bbF,\, 1 \le i \le r$, we see that the 
equality 
$AB = I_r$ directly translates into the desired result.
\end{proof}

\begin{lemma}\label{lem:quotient-space-dimension-long}
Let $\bfx_1, \ldots, \bfx_r \in \bbF^\calI$. Consider the linear map $\Phi \colon \bigoplus_\calI \bbF \to \bbF^r,\, h \mapsto ( \bfx_1(h), \ldots, \bfx_r(h) )$. Then $\dim (\bigoplus_\calI \bbF / \ker \Phi) = \dim \Span \{ \bfx_i \}_{i = 1}^r$.
\end{lemma}
\begin{proof}
By the First Isomorphism Theorem for vector spaces $\bigoplus_I \bbF / \ker \Phi \cong \im \Phi$.
So it suffices to prove $\dim \im \Phi = \dim \Span \{ \bfx_i \}_{i = 1}^r$. 
Clearly it suffices to prove the case when
$\bfx_1, \ldots, \bfx_r$ are linearly independent, and that follows
directly from Lemma~\ref{lem:linear-independence-and-dual-vectors-long}.
\end{proof}

\begin{lemma}\label{lem:tensor-powers-linear-independence-nonsymmetric-rank-and-uniqueness-long}
Let $r \ge 0$, $n \ge 2$ and let $\bfx_{1, j}, \ldots, \bfx_{r, j} \in \bbF^{\calI_j}$ be $r$ linearly independent vectors for $1 \le j \le n$ and $a_1, \ldots, a_r \in \bbF \setminus \{0\}$.
Then the tensor
\begin{equation}\label{eq:tensor-powers-linear-independence-nonsymmetric-rank-and-uniqueness-long}
A = \sum_{i = 1}^r a_i \bfx_{i, 1} \otimes \cdots \otimes \bfx_{i, n} 
\end{equation}
has $\rank(A) = r$. For $n \ge 3$, any expression of $A$ as $\sum_{i = 1}^r b_i \bfy_{i, 1} \otimes \cdots \otimes \bfy_{i, n}$ is a permutation of the sum in (\ref{eq:tensor-powers-linear-independence-nonsymmetric-rank-and-uniqueness-long}).
\end{lemma}
\begin{proof}
When $r = 0$, the statement is trivially true as $A =0$, so we assume $r \ge 1$.
Let $n \ge 2$ and let $\rank(A) = s$. Clearly $s \le r$.
By the definition of having tensor rank $s$, there
exist $\bfy_{1, j}, \ldots, \bfy_{s, j} \in \bbF^{\calI_j}$
for each $1 \le j \le n$
and $b_1, \ldots, b_s \in \bbF \setminus \{0\}$ such that
\begin{equation}\label{eqn:tensor-powers-linear-independence-nonsymmetric-rank-and-uniqueness-lemma-two-sided-decomposition-long}
\sum_{i = 1}^r a_i \bfx_{i, 1} \otimes \cdots \otimes \bfx_{i, n} = A = \sum_{j = 1}^s b_j \bfy_{j, 1} \otimes \cdots \otimes \bfy_{j, n}.
\end{equation}
By Lemma~\ref{lem:linear-independence-and-dual-vectors-long},
there exist $h_{1, j}, \ldots, h_{r, j}$ dual to
$\bfx_{1, j}, \ldots, \bfx_{r, j}$ for each $1 \le j \le n$.
For any $1 \le i \le r$,
applying $h_{i, 2} \otimes \cdots \otimes h_{i, n}$ to the sum,
where $h_{i, k}$ is applied along the coordinate space $\bbF^{\calI_k}$
 for $2 \le k \le n$, 
we get $a_i \bfx_{i, 1}$ as a linear combination
of $\bfy_{1, 1}, \ldots, \bfy_{s, 1}$. 
As $a_i \not =0$, and $\bfx_{1, 1}, \ldots, \bfx_{r, 1}$ are linearly independent,
we get $s \ge r$.
So $s = r$,
and $\bfy_{1, 1}, \ldots, \bfy_{s, 1}$ are linearly independent.
Analogously,
we get that
$\bfy_{1, j}, \ldots, \bfy_{s, j}$
are linearly independent for $2 \le j \le n$.

Next, let $n \ge 3$ and consider 
(\ref{eqn:tensor-powers-linear-independence-nonsymmetric-rank-and-uniqueness-lemma-two-sided-decomposition-long}) again, where $s = r$.
For any $1 \le i \le r$,
applying $h_{i, n}$ along  the coordinate space $\bbF^{\calI_n}$,
we get
\begin{equation}\label{eqn:tensor-powers-linear-independence-nonsymmetric-rank-and-uniqueness-lemma-expression-of-ai-otimes-x(i,1)-to-x(i,n-1)-alt-long}
a_i \bfx_{i, 1} \otimes \cdots \otimes \bfx_{i, n - 1} = B = \sum_{j = 1}^r b_j \bfy_{j, n}(h_{i, n}) \bfy_{j, 1} \otimes \cdots \otimes \bfy_{j, n - 1}.
\end{equation}
From the LHS, $\rank(B) = 1$.
 By applying what has just been proved to $B$
having $\rank(B) = 1$, and $\bfy_{1, k}, \ldots, \bfy_{r, k}$
are linearly independent for all $1 \le k \le n-1$, 
from the expression of $B$ in the RHS of 
(\ref{eqn:tensor-powers-linear-independence-nonsymmetric-rank-and-uniqueness-lemma-expression-of-ai-otimes-x(i,1)-to-x(i,n-1)-alt-long})
 we see that  the number of
terms with nonzero coefficients on the RHS is exactly one.
 Since $b_j \not = 0$,
it follows that for any $i$, there is exactly one $j = j_n(i)$ such that
 $\bfy_{j, n}(h_{i, n}) \not = 0$, and 
(\ref{eqn:tensor-powers-linear-independence-nonsymmetric-rank-and-uniqueness-lemma-expression-of-ai-otimes-x(i,1)-to-x(i,n-1)-alt-long}) takes the form
\begin{equation}\label{eqn:one-term-tensorpower}
a_i \bfx_{i, 1} \otimes \cdots \otimes \bfx_{i, n - 1} = B = 
b_j \bfy_{j, n}(h_{i, n}) \bfy_{j, 1} \otimes \cdots \otimes \bfy_{j, n - 1},
\end{equation}
where $j = j_n(i)$.
For any $1 \le k \le n-1$, 
applying $\bigotimes_{1 \le s \le n-1: s \not = k} h_{i, s}$ to $B$,
where $h_{i, s}$ is  applied along  $\bbF^{\calI_s}$,
we get
$a_i \bfx_{i, k} =  b'_j \bfy_{j, k}$,
where $1 \le i \le r$, $j = j_n(i)$,
 and $b'_j = b_j \prod_{1 \le s \le n: s \not = k} 
\bfy_{j, s}(h_{i,s})$.
Since $a_i \ne 0$ and
 $\bfx_{1, k}, \ldots, \bfx_{r, k}$ are linearly independent,
we conclude that
the map $i \mapsto j = j_n(i)$ is a permutation, and  $b'_j \ne 0$.
%
%

Similarly for any $1 \le i \le r$, we can apply 
$h_{i, 1}$ along  the coordinate space $\bbF^{\calI_1}$ to
(\ref{eqn:tensor-powers-linear-independence-nonsymmetric-rank-and-uniqueness-lemma-two-sided-decomposition-long}),
and get
\begin{equation}\label{eqn:another B-eqn}
a_i \bfx_{i, 2} \otimes \cdots \otimes \bfx_{i, n} = B' = \sum_{j = 1}^r b_j \bfy_{j, 1}(h_{i, 1}) \bfy_{j, 2} \otimes \cdots \otimes \bfy_{j, n}.
\end{equation}
By the same argument, 
there is exactly one $j = j_1(i)$ such that
for all $2 \le k \le n$, 
\begin{equation}\label{eqn:x-through-y}
a_i \bfx_{i, k}  = b_j
\bigg( \prod_{\text{\small $\substack{1 \le s \le n \\ 
s \not = k}$}} \bfy_{j, s}(h_{i, s}) \bigg)
 \bfy_{j, k}.
\end{equation}
Taking $k=2$, by the linear independence of $\bfx_{1, 2}, \ldots, \bfx_{r, 2}$,
we have $j_1(i) = j_n(i)$, for all $1 \le i \le r$, which we now denote simply
by $j(i)$.

We also have $0 \ne a_i = b_j \prod_{s=1}^n \bfy_{j, s}(h_{i, s})$,
which implies that for all $1 \le k \le n$,
\[\bfy_{j, k} = \frac{a_i}{b_j \prod_{\text{\small $\substack{1 \le s \le n \\
s \not = k}$}} \bfy_{j, s}(h_{i, s})} \bfx_{i, k} 
= \bfy_{j, k}(h_{i, k})  \bfx_{i, k},\]
where $j = j(i)$.
It follows that, for all $1 \le i \le r$, and $j =j(i)$,
\[b_j \bfy_{j, 1} \otimes \cdots \otimes \bfy_{j, n}
= b_j  \prod_{s=1}^n \bfy_{j, s}(h_{i, s}) 
\bfx_{i, 1} \otimes \cdots \otimes \bfx_{i, n}
= a_i \bfx_{i, 1} \otimes \cdots \otimes \bfx_{i, n}.\]
\end{proof}

\begin{corollary}\label{cor:tensor-powers-linear-independence-nonsymmetric-and-symmetric-rank-and-uniqueness-long}
Let $r \ge 0$, and
let $\bfx_1, \ldots, \bfx_r \in \bbF^\calI$ be $r$ linearly independent vectors and $a_1, \ldots, a_r \in \bbF \setminus \{0\}$. Then for any integer $n \ge 2$, the symmetric tensor
\begin{equation}\label{eq:tensor-powers-linear-independence-nonsymmetric-and-symmetric-rank-and-uniqueness-long}
A = \sum_{i = 1}^r a_i \bfx_i^{\otimes n} \in \Sym^n(\bbF^\calI)
\end{equation}
has $\rank(A) = \symrank(A) = r$. 
\end{corollary}
\begin{proof}
Clearly, $\rank(A) \le \symrank(A) \le r$. By Lemma~\ref{lem:tensor-powers-linear-independence-nonsymmetric-rank-and-uniqueness-long}, $\rank(A) = r$
so $\symrank(A) = r$. 
\end{proof}

\section*{Acknowledgements}
We thank Guus Regts for telling us
about his beautiful work,
and for very insightful comments on this paper. We also thank 
the anonymous referees of the  11th  
\emph{Innovations  in  Theoretical  Computer  Science}  (ITCS) conference
2020 for their helpful comments.

\bibliography{references}

\end{document}